\documentclass[11pt,letterpaper]{article}

\usepackage{paperlighter}

\usepackage{microtype}
\usepackage{graphicx}
\usepackage{subfigure}
\usepackage{booktabs} 
\usepackage{awesomebox}
\usepackage[ruled,vlined,resetcount,linesnumbered]{algorithm2e}
\makeatletter
\newcommand{\nosemic}{\renewcommand{\@endalgocfline}{\relax}}
\newcommand{\dosemic}{\renewcommand{\@endalgocfline}{\algocf@endline}}
\makeatother
\usepackage{amsmath}
\usepackage{amssymb}
\usepackage{mathtools}
\usepackage{amsthm}

\usepackage[capitalize,noabbrev]{cleveref}

\theoremstyle{plain}
\newtheorem{theorem}{Theorem}[section]

\newtheorem{lemma}[theorem]{Lemma}
\newtheorem{claim}[theorem]{Claim}

\theoremstyle{definition}
\newtheorem{definition}[theorem]{Definition}

\theoremstyle{remark}

\usepackage[textsize=tiny]{todonotes}

\slimauthor{Siddhartha Sarkar}

\begin{document}

\lightertitle{Faster Algorithm for Minimum Ply Covering of Points with Unit Squares}
\lighterauthor{Siddhartha Sarkar}
\lighteraffil{Computer Science and Automation, Indian Institute of Science, Bengaluru}
\lighteremail{siddharthas1@iisc.ac.in}

\begin{abstract}
Biedl et al. introduced the minimum ply cover problem in CG 2021 following the seminal work of Erlebach and van Leeuwen in SODA 2008. They showed that determining the minimum ply cover number for a given set of points by a given set of axis-parallel unit squares is NP-hard, and gave a polynomial time $2$-approximation algorithm for instances in which the minimum ply cover number is bounded by a constant. Durocher et al. recently presented a polynomial time $(8 + \epsilon)$-approximation algorithm for the general case when the minimum ply cover number is $\omega(1)$, for every fixed $\epsilon > 0$. They divide the problem into subproblems by using a standard grid decomposition technique. They have designed an involved dynamic programming scheme to solve the subproblem where each subproblem is defined by a unit side length square gridcell. Then they merge the solutions of the subproblems to obtain the final ply cover. We use a horizontal slab decomposition technique to divide the problem into subproblems. Our algorithm uses a simple greedy heuristic to obtain a $(27+\epsilon)$-approximation algorithm for the general problem, for a small constant $\epsilon>0$. Our algorithm runs considerably faster than the algorithm of Durocher et al. We also give a fast $2$-approximation algorithm for the special case where the input squares are intersected by a horizontal line. The hardness of this special case is still open. Our algorithm is potentially extendable to minimum ply covering with other geometric objects such as unit disks, identical rectangles etc.
\end{abstract}
\section{Introduction}
Set Cover is a fundamental problem in combinatorial optimization. Given a range space $(X, \mathcal{R})$ consisting of a set $X$ and a family $\mathcal{R}$ of subsets of $X$ called the ranges, the goal is to compute a minimum cardinality subset of $\mathcal{R}$ that covers all the points of $X$. It is NP-hard to approximate the minimum set cover within a logarithmic factor \cite{Raz_setcover_hard_97, feige_1998}. When the ranges are derived from geometric objects, it is called the Geometric Set Cover problem. Computing the minimum cardinality set cover remains NP-hard even for simple 2D objects, such as unit squares on the plane \cite{geom_set_cover_2D_hardness}. There is a rich literature on designing approximation algorithms for various geometric set cover problems (see \cite{GeomSetCover_agarwal_pan_2014, Clarkson_Varadarajan_geomsetcover, Hochbaum_Maass_covering, Chan_Grant_pack_cover, raman_geom_setcover_focs2014, Erlebach_wted_set_cover}). More often than not, the geometric versions of the covering problems are efficiently solvable or approximated well. Many variants of the Geometric Set Cover problem find applications in facility location, interference minimization in wireless networks, VLSI design, etc \cite{Demaine_application_of_coverage, geomsetcover_wireless}. In this paper, we investigate the following covering problem.

\textit{Problem Statement.} Given a set of geometric objects $\mathcal{S}$, the \textit{ply} of $\mathcal{S}$ is the maximum cardinality of any subset of $\mathcal{S}$ that has non-empty common intersection. A set of objects $\mathcal{S}$ is said to cover a set of points $P$ if for each point $p\in P$ there exists an object $S\in \mathcal{S}$ such that $p$ is contained in $S$. Given a set of objects (for e.g., unit squares) $\mathcal{S}$, and a set of points $P$ on the plane, the goal is to pick a subset $\mathcal{S'}\subseteq \mathcal{S}$ such that every point in $P$ is covered by at least an object in $\mathcal{S'}$ while minimizing the ply.

In this paper we will deal with the minimum ply cover problem where the objects to cover with are unit side length squares.

\subsection{Our contribution} We design a simple greedy heuristic for axis-parallel unit squares on the plane that achieves an approximation factor of $(27+\epsilon)$, where $\epsilon>0$ is a small constant. Our algorithm runs in $O(n^2m^2)$ time where $n$ is the number of input points and $m$ is the number of input squares. Our algorithm is considerably faster when compared to the DP algorithm of \cite{Durocher_Mondal_Minply} that requires $O(n + m^{8}k^{4}\log k + m^{8}\log m\log k)$ time, where $k$ is the optimal ply value. Our algorithm is easy to implement and extend to other covering objects.

\subsection{Related Work}
The minimum ply cover problem is a generalization of the minimum membership set cover (i.e., MMSC) problem. Given a set of points $P$ and a family $\mathcal{S}$ of subsets of $P$, the goal in MMSC is to find a subset $\mathcal{S'}\subseteq \mathcal{S}$ that covers $P$ while minimizing the maximum number of elements of $\mathcal{S'}$ that contain a common point of $P$. \cite{kuhn_mmsc} introduced the general problem and they presented an LP based technique to obtain an approximation factor of $\ln n$ which matches the hardness lower bound. The membership is measured at the points of $P$ in the MMSC problem. \cite{erlebach_soda_minply} proved that the minimum membership set cover of points with unit disks or unit squares is NP-hard and cannot be approximated by a ratio smaller than $2$. For unit squares, they present a $5$-approximation algorithm with the assumption that the minimum ply value is bounded by a constant. In the minimum ply cover (abbr. MPC) problem, the ply (i.e., membership) is measured at any point on the plane. \cite{Biedl_Minply} prove that the MPC problem is NP-hard for both unit squares and unit disks, and does not admit polynomial-time approximation algorithms with ratio smaller than $2$ unless P=NP. They gave polynomial time $2$-approximation algorithms for this problem with respect to unit squares and unit disks, when the minimum ply value is bounded by a constant.

\section{Minimum Ply Covering}
Our approach is to divide the problem into distinct subproblems and solve the subproblems. Then we merge the solutions to obtain the ply cover of the original problem. First, let us define and solve a special case of the minimum ply cover problem. The techniques used here will be useful in the future.

\subsection{Squares are intersected by a horizontal line}
Suppose that the input squares are intersected by a horizontal line and the input points lie in at least one of the input squares. Refer to Figure (\ref{fig:mpcsihl}). For this special case, we give an algorithm that computes the minimum ply cover of ply at most twice the optimal ply. The algorithm is simple to state and analyze. First, we divide the problem into two subproblems. One corresponding to the input points above the horizontal line and the other corresponding to the input points below the horizontal line.

\begin{figure}[ht!]
\centering
\includegraphics[width=8cm]{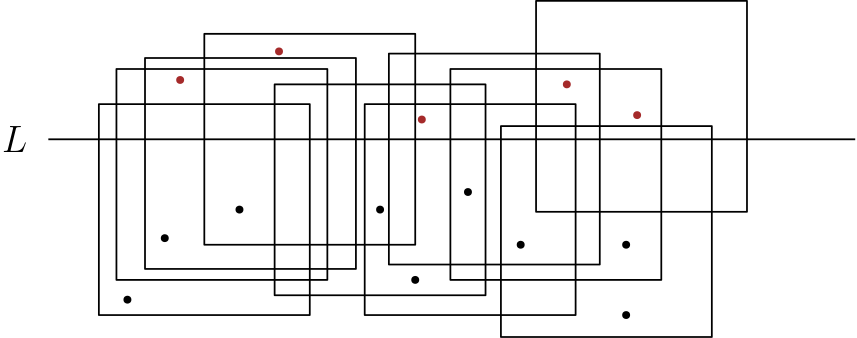}
\caption{An instance of the special case where all the input squares are intersected by a horizontal line $L$ and the points lie on both sides.}
\label{fig:mpcsihl}
\end{figure}



\begin{definition}\label{defn:algo_mpsc_horiline}
Consider a horizontal line $L$. Now consider a set $\mathcal{S}$ of unit squares intersecting $L$. Let $P$ be a set of points located below $L$ such that each point lies inside at least one of the squares in $\mathcal{S}$. Then \textbf{MPCSIHL1} is the problem of computing the minimum ply cover of $P$ with the squares in $\mathcal{S}$.
\end{definition}
\textbf{MPCSIHL1} stands for \textbf{M}inimum \textbf{P}ly \textbf{C}over for unit \textbf{S}quares \textbf{I}ntersected by a \textbf{H}orizontal \textbf{L}ine where points lie on only \textbf{1} side (i.e., above or below the horizontal line). Refer to Figure (\ref{fig:mpcsihl1})

\begin{figure}[ht!]
\centering
\includegraphics[width=8cm]{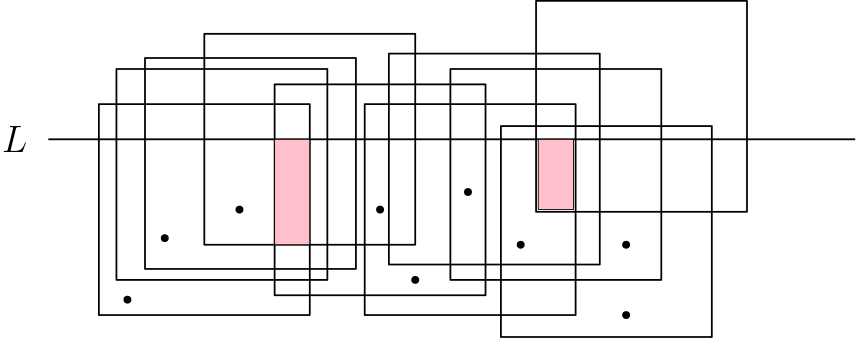}
\caption{An instance of the \textbf{MPCSIHL1} problem where all the input points lie below $L$. The pink regions are ply regions with ply value $5$.}
\label{fig:mpcsihl1}
\end{figure}

For a set of squares $\mathcal{S}$, we call a maximum \textit{depth} region on the plane as the \textit{ply region} and denote it by $clique(\mathcal{S})$. So, the \textit{ply region} is the common intersection region below the line $L$. The corresponding depth is the ply which is denoted by $ply(\mathcal{S})$. Note that the \textit{ply region} of a set of squares $\mathcal{S}$ may not be unique. Refer to Figure \ref{fig:mpcsihl1}. The Algorithm (\ref{algo:exact_dp}) is a heuristic that computes a minimum ply cover incrementally. The heuristic uses a greedy rule. For multiple sets of squares $S_1, S_2, \ldots, S_{k}$, we define the following \textit{greedy criteria} for deciding which one to prefer.

\textit{Greedy criteria}. If a set of squares has multiple ply regions, consider the rightmost one as its representative. Use following preference rules to for selection and breaking ties.
\begin{enumerate}
    \item Prefer the set of squares having the minimum ply value.
    \item If rule (1) leads to a tie, prefer the set of squares with the leftmost \textit{ply region} right side.
    \item If rule (2) also leads to a tie, prefer the set of squares with the narrowest \textit{ply region}. 
    \item If the tie persists, then select a candidate set of squares arbitrarily.
\end{enumerate}

\begin{algorithm}[H]
\SetAlgoLined
\SetKwInOut{Input}{Input}
\SetKwInOut{Output}{Output}
\Input{ A horizontal line $L$, a set $\mathcal{S}$ of $m$ unit squares  intersecting $L$, a set $P$ of $n$ points below $L$ such that each of them lies in at least one square in $\mathcal{S}$.}
\Output{ Returns a set of squares $\mathcal{S'}\subseteq \mathcal{S}$ covering $P$ while minimizing ply.}
\nosemic 
Sort the points in $P$ from left to right.\;
Let the sorted order of points be $p_1, p_2, \ldots, p_n$. \;
$T\leftarrow \infty$\hspace{5mm}//  A 2D array $T$ of dimension $n\times m$ where each entry is  initialized to $\infty$ or \textit{infeasible}.\;
\For{$i\gets1$ \KwTo $n$}{
    \For{$j\gets1$ \KwTo $m$}{
        $T[i, j] = Compute\-Entry(i, j, T)$\;
    }
}
Obtain the best entry in the $n$-th row as per the \textit{greedy criteria}, say $Sol$.\;
\Return $Sol$
\caption{Exact algorithm for MPCSIHL1}
\label{algo:exact_dp}
\end{algorithm}
\vspace{5mm}
The procedure Compute\-Entry($\cdot$) shows the computation of each entry in the table. Observe that the entry $T[i, j]$ takes into account all the feasible solutions in the $(i-1)$-th row of the table $T$. If $T[i, j]$ is formed by taking the union of $T[i-1, k]$ with $s_{j}$, for some $k\in [m]$, then we say that $T[i-1, k]$ is the parent of $T[i, j]$ and $T[i, j]$ is derived from $T[i-1, k]$. Intuitively, $T[i, j]$ represents the minimum ply cover for $P_{i}=\{p_1, \ldots, p_i\}$ as per the \textit{greedy criteria} such that the point $p_{i}$ is covered by the square $s_j\in T[i, j]$. Starting from $Sol$, it is possible to trace a path to an entry in the first row of the table by following the parents. The path has one node from each row of the table $T$.

After the computation is over, for the $i$-th row in $T$ corresponding to the point $p_{i}$, there can be multiple entries/covers that achieve the same ply. In the next iteration, while computing $T[i+1, j]$, we will use the \textit{greedy criteria} to select the best entry from the $i$-th row.

Denote by $P_i$ the set $\{p_1, p_2, \ldots, p_i\}$ of the $i$ leftmost input points. The entry $T[i, j]$ is a feasible ply cover for $P_{i}$, with the constraint that $s_j$ is included in the solution. The algorithm (\ref{algo:exact_dp}) computes a set cover $Sol$ for $P$. We claim that our solution is optimal. We will prove the following \textbf{loop invariant} to argue the correctness of our claim. 

\begin{claim}\label{claim:loop_invariant}
(\textit{loop invariant}) At the end of the $i$-th iteration of the outer \textit{for} loop at line $4$ of Algorithm (\ref{algo:exact_dp}), in each row $l$ such that $1\leq l\leq i$, there exists an entry $T(l, j)$ corresponding to the minimum ply set cover for the set of points $P_l$ such that it is the best as per the \textit{greedy criteria}; $j$ can vary for each row.
\end{claim}

\begin{algorithm}[H]
\SetKwInput{KwInput}{Input}                
\SetKwInput{KwOutput}{Output}              
\DontPrintSemicolon
  

  \SetKwFunction{FSum}{Compute\-Entry}
  
  \SetKwProg{Fn}{Function}{:}{}
  \Fn{\FSum{$i$, $j$, $T$}}{
        \If{$i==1$}{
            \If{$p_{i}\in s_j$}{
                \Return $\{s_j\}$
            }
        \Else{
        \Return $\infty$
        }
        }
        Let $\mathcal{F}_{i-1}$ be the set of feasible solutions in the $(i-1)$-th row, i.e., the entries with finite value.\;
        \For{$F\in \mathcal{F}_{i-1}$}{
        Compute the ply region of $F\cup \{s_j\}$ \;
        }
        Select the best entry among the solutions computed in the for loop in line $9$, as per the \textit{greedy criteria}, say $Sol_{ij}$.\;
        \KwRet $Sol_{ij}$ \;
  }
\end{algorithm}

\begin{proof}
First, we consider the base case. For the point $p_1$, our algorithm fills the first row of $T$ with feasible solutions containing one square each, i.e., for each square $s_j$ covering $p_1$, the corresponding table entry $T(1, j)$ is a singleton set $\{s_j\}$. The ply of each feasible solution is $1$, which is trivially optimal. The entry in row $1$ corresponding to the leftmost square containing $p_1$ is the optimal solution for $P_1$ based on our \textit{greedy criteria}.

We fix an $1< i< n$. Our inductive assumption is that the algorithm has already computed the optimal solution for the first $(i-1)$ iterations of the outer \textit{for} loop. Specifically, we assume that our \textit{loop invariant} is true after the first $(i-1)$ iterations.

We need to prove that at the end of the $i$-th iteration of the outer \textit{for} loop, there exists a square $s_j$ containing the point $p_{i}$ such that the entry $T[i, j]$ is the \textit{optimal} ply cover for $P_{i}$ satisfying the \textit{greedy criteria}. We term such a solution as the optimal \textit{parametric} solution.

After the $i$-th iteration of the outer \textit{for} loop, let $Sol_{i}$ denote an optimal set cover (in the sense of the \textit{greedy criteria}) stored at the $i$-th row of the table $T$, for $1\leq i\leq n$. Suppose for the sake of contradiction, there exists a better solution $OPT_{ij^{*}}$ for $P_i$ and $p_{i}\in s_{j^{*}}$ and $s_{j^{*}}\in OPT_{ij^{*}}$. That means either
\begin{align}\label{eqn:1a}
    ply(OPT_{ij^{*}})< ply(Sol_{i})
\end{align}
Or 
\begin{align}\label{eqn:1b}
    ply(OPT_{ij^{*}}) = ply(Sol_{i}),\hspace{5mm} \text{but $OPT_{ij^{*}}$ is better than $Sol_{i}$ as per the \textit{greedy criteria}.}
\end{align}

Let's consider the possibility (\ref{eqn:1a}) first. Our algorithm essentially constructs $Sol_{i}$ by combining a square $s_j$ with a feasible solution $T(i-1, j')$ for $P_{i-1}$. If $s_j\in T(i-1, j')$ or $s_j$ does not intersect the ply region of $T(i-1, j')$, then $T(i-1, j')\cup \{s_j\}$ is a feasible solution for $P_i$ with ply $ply(T(i-1, j'))$.

Let $OPT_{(i-1)j''}\subseteq OPT_{ij^{*}}$ be a set cover (we do not claim that it has minimum ply) for $P_{i-1}$ with no redundant squares. One way to construct $OPT_{(i-1)j''}$ from $OPT_{ij^{*}}$ is to simply remove the square (if any) that covers $p_i$ exclusively  from $OPT_{ij^{*}}$. Let $s_{j''}\in OPT_{(i-1)j''}$ and $p_{i-1}\in s_{j''}$. Because of the subset relationship, we have
\begin{align}\label{eqn:2}
ply(OPT_{(i-1)j''}) \leq ply(OPT_{ij^{*}})
\end{align}

By our inductive hypothesis, we have
\begin{align}\label{eqn:3}
ply(Sol_{i-1}) \leq ply(OPT_{(i-1)j''})
\end{align}
Since our algorithm adds at most one square for every new point, we have
\begin{align}\label{eqn:4}
ply(Sol_{i}) \leq 1 + ply(Sol_{i-1})
\end{align}
Therefore, from (\ref{eqn:1a}), (\ref{eqn:2}), (\ref{eqn:3}), and (\ref{eqn:4}), we conclude that
\[
ply(OPT_{(i-1)j''}) = ply(OPT_{ij^{*}}) = ply(Sol_{i-1}) = ply(Sol_{i}) - 1
\]
Now consider the solutions $OPT_{(i-1)j''}$ for $P_{i-1}$ and $OPT_{ij^{*}}$ for $P_{i}$. There are two ways in which $OPT_{ij^{*}}$ can be related to $OPT_{(i-1)j''}$. 

\textbf{Case 1}: $\exists$ a square $s_{j^{*}}\in OPT_{ij^{*}}\setminus OPT_{(i-1)j''}$ covering $p_{i}$ such that $s_{j^{*}}$ does not intersect the ply region of $OPT_{(i-1)j''}$. 

\textbf{Case 2}: $OPT_{(i-1)j''}$ already covers the point $p_i$ and thus $OPT_{ij^{*}} = OPT_{(i-1)j''}$. 

Figure (\ref{fig:img1}) shows the two cases.
\begin{figure}[ht!]
\centering
\includegraphics[width=13cm]{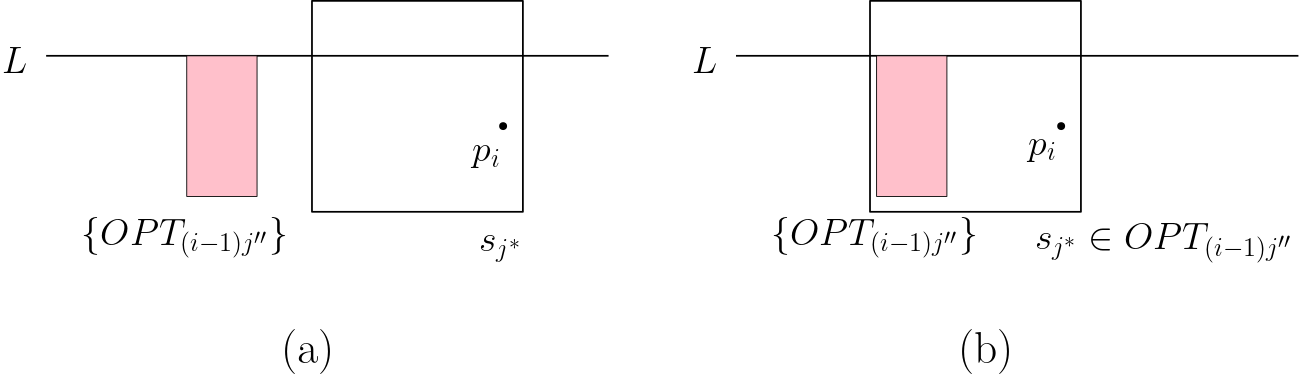}
\caption{The two cases for $OPT_{(i-1)j''}$ with respect to covering the point $p_{i}$. (a) Case 1, (b) Case 2. The pink region $\{OPT_{(i-1)j''}\}$ refers to the ply region of $OPT_{(i-1)j''}$.}
\label{fig:img1}
\end{figure}

First we consider \textbf{Case 1}. Consider a square $s_{j''}\in OPT_{(i-1)j''}$ that covers the point $p_{i-1}$ in $OPT_{(i-1)j''}$. By our inductive hypothesis,
\[
ply(Sol_{i-1}) \leq ply(OPT_{(i-1)j''})
\]
The square $s_{j^{*}}$ cannot lie entirely to the left of the left boundary of the rightmost ply region of $OPT_{(i-1)j''}$. Suppose it does. The left boundary of the ply region of $OPT_{(i-1)j''}$ is determined by the rightmost square, say $s_r$, participating in the ply. Since, no square in $OPT_{(i-1)j''}$ is redundant, $s_r$ must cover a point, say $q\in P_{i-1}$, exclusively. Since, by our assumption, $s_{j^*}$ lies to the left of the left boundary of $s_r$, and $p_{i}$ lies to the right of $q$, hence $s_{j^*}$ cannot cover $p_{i}$. Hence we have a contradiction. Refer to Figure (\ref{fig:contradiction}) for an illustration. 

\begin{figure}[ht!]
\centering
\includegraphics[width=6cm]{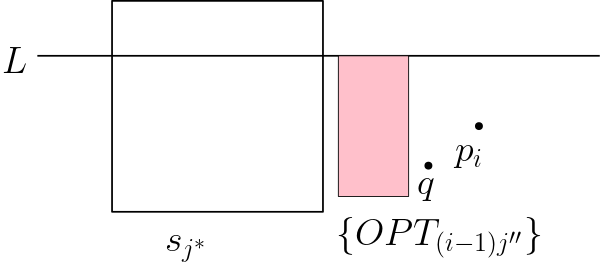}
\caption{The square $s_{j^*} \ni p_i$ lying entirely to the left of the ply region of $OPT_{(i-1)j''}$ is not possible.}
\label{fig:contradiction}
\end{figure}

Therefore, $s_{j^*}$ must lie entirely to the right of the right boundary of the rightmost ply region of $OPT_{(i-1)j''}$. By virtue of our inductive assumption (about the \textit{greedy criteria} obeyed at every iteration), $Sol_{i-1}\cup \{s_j\}$ is a feasible solution of $P_{i}$ with ply equal to $ply(Sol_{i-1})$. Therefore,
\[
ply(Sol_{i}) = ply(Sol_{i-1}) \leq ply(OPT_{(i-1)j''}) = ply(OPT_{ij^{*}})
\]
We have arrived at a contradiction. Therefore, \textbf{Case 1} is ruled out. 

Now we consider \textbf{Case 2} where $p_i$ is already covered by $OPT_{(i-1)j''}$. For a square $s$, we write $x_{L}(s), x_{R}(s)$ to denote the $x$-coordinates of the left and right boundaries of $s$ respectively. Similarly, write $y_{T}(s), y_{B}(s)$ to denote the $y$-coordinates of the top and bottom boundaries of $s$, respectively. We assume that no two squares share the same left boundary. More specifically, the input squares have a left-to-right ordering $\prec$. We write $s_1\prec s_2$ if $s_1$ lies to the left of $s_2$, i.e., $x_{L}(s_1)<x_{L}(s_2)$. If $i<j$, then $p_{i}$ lies to the left of $p_{j}$, i.e., $x(p_{i})<x(p_{j})$.

Let $s_k$ be the rightmost square among the squares covering $p_{i}$ in $OPT_{(i-1)j''}$. Since $s_k\in OPT_{(i-1)j''}$, there must exist at least one point to the left of $p_i$ which is exclusively covered by $s_k$ in $OPT_{(i-1)j''}$.

Let $p_{i'}$, where $i'<i$, be the leftmost point among the points exclusively covered by $s_k$ in $OPT_{(i-1)j''}$. Since none of the leftmost $(i'-1)$ points is exclusive to $s_k$; hence $OPT_{i'}\setminus OPT_{i'-1}=\{s_k\}$. If $s_k$ contributes to the ply, i.e., $s_k$ intersects the ply region of $OPT_{i'-1}$,  then $[OPT_{i'}] = [OPT_{i'-1}]+1$, else $[OPT_{i'}] = [OPT_{i'-1}]$. 

\begin{claim}
The square $s_k$ is the rightmost square in $OPT_{(i-1)j''}$ as per the left to right ordering $\prec$.
\end{claim}
\begin{proof}
Suppose there exists a square $s_o\in OPT_{(i-1)j''}$ which is rightwards with respect to $s_k$, i.e., $s_k \prec s_o$. The left boundary of $s_o$ must lie to the left of $p_{i-1}$, otherwise $s_o$ will become redundant. If $y_{T}(s_o) > y_{T}(s_k)$, then $s_o$ becomes redundant in $OPT_{(i-1)j''}$ as $s_k$ covers the relevant area of $s_o$. The other possibility is $y_{T}(s_o) < y_{T}(s_k)$. But the square $s_o$ will contain $p_i$ in this case. This will violate the definition of $s_k$, i.e., $s_k$ being the rightmost among the squares containing $p_i$ in $OPT_{(i-1)j''}$. Hence this is ruled out. Therefore, $s_k$ is the rightmost square in $OPT_{(i-1)j''}$. Refer to Figure (\ref{fig:img2}).
\end{proof}

\begin{figure}[ht!]
\centering
\includegraphics[width=12cm]{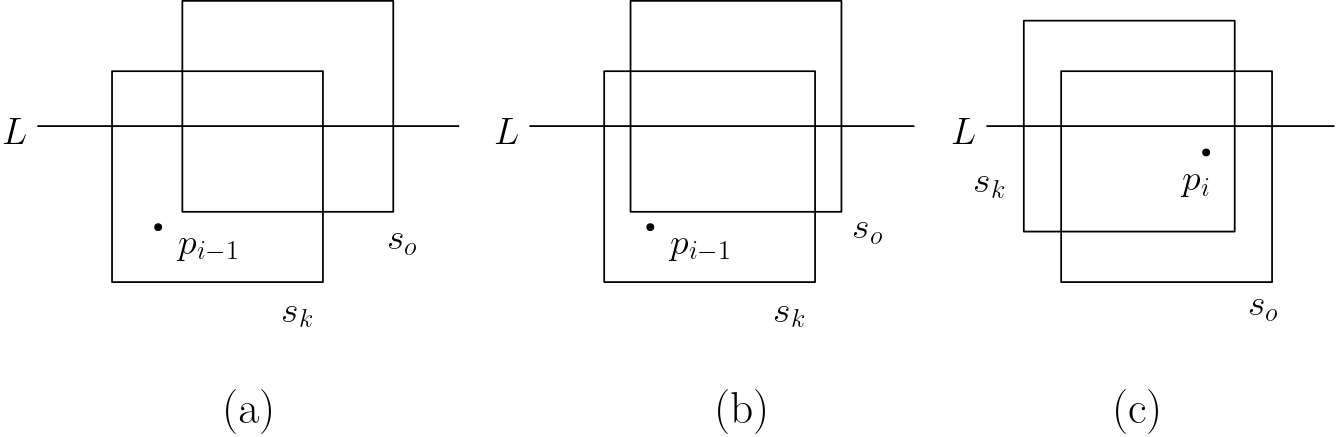}
\caption{Relative positions of $s_{k}$ and $s_{o}$. (a) $s_o$ is above $s_{k}$ but $p_{i-1}$ is to the left of the left boundary of $s_o$. This makes $s_o$ redundant. (b) $s_o$ is above $s_{k}$ but the entire area of $s_o$ to the left of $p_{i-1}$ is covered by $s_k$. This makes $s_o$ redundant. (c) $s_o$ is below $s_{k}$, here $s_o$ also covers $p_i$ contradicting the definition of $s_{k}$. Thus $s_k$ is the rightmost square in $OPT_{(i-1)j''}$.}
\label{fig:img2}
\end{figure}
\begin{claim}\label{claim: contains}
All points between $p_{i'}$ and $p_{i}$ must lie in the square $s_k$.
\end{claim}
\begin{proof}
Let there be a point $p_r$ such that $i'<r<i$ and $p_r \notin s_k$. Essentially, $p_r$ lies below the bottom boundary of $s_k$. There must exist a square $s\in OPT_{(i-1)j''}$ covering $p_r$. By our previous claim, $s_k$ is the rightmost square in $OPT_{(i-1)j''}$. Hence $s\prec s_k$. Therefore, $s$ must also cover $p_{i'}$ which is a contradiction since $p_{i'}$ is exclusively covered by $s_k$ in $OPT_{(i-1)j''}$. 
\end{proof}
\begin{claim}
The entry $T(i', k)$ is a feasible set cover for $P_{i}$ and is as good as the solution $OPT_{ij^{*}}$. Mathematically, $ply(T(i', k)) = ply(OPT_{ij^{*}})$.
\end{claim}
\begin{proof}
By the inductive hypothesis, $ply(Sol_{i'}) \leq ply(OPT_{i'})$. Since $[Sol_{i'}]$ is the optimal ply for covering $P_{i'}$ and due to the definition of our function Compute\_Entry($\cdot$), we have 
\begin{align}\label{eq:6}
    [Sol_{i'}] \leq [T(i',k)]\leq [Sol_{i'}]+1
\end{align}

Hence we need to consider only two possible cases: i) $[T(i', k)] = [Sol_{i'}]$ and ii) $[T(i', k)] = [Sol_{i'}] + 1$.

i) $[T(i', k)] = [Sol_{i'}]$: Since all points between $p_{i'}$ and $p_i$ lie in $s_k$ due to Claim (\ref{claim: contains}), $T(i', k)$ is a feasible solution for $P_i$. If we picked $s_k$, i.e., $T(i', k)$ for $p_{i'}$ then we are done. If we did not pick $s_k$, then say we picked $s_{k'}$ that does not contain $p_{i}$. The entry $T(i', k)$ will continue to be present in the table till the $(i-1)$-th row. Eventually, in some $(i'')$-th iteration, the algorithm will select $T(i', k)$ since it covers $p_i$ and does not lead to an increase in ply, where $i'\leq i''\leq i$. Thus we have a feasible solution $T(i', k)$ in the $(i-1)$-th row of our table till the $i$-th row and $[T(i', k)]\leq [OPT_{i}]$. Hence $[Sol_{i}] > [OPT_{ij^{*}}]$ is ruled out. 

ii) $[T(i', k)] = [Sol_{i'}] + 1$: Since $[T(i', k)] \leq [Sol_{i'-1}] + 1$, we have $[Sol_{i'}] = [Sol_{i'-1}]$. In other words, picking $s_k$ for covering $p_{i'}$ will lead to an increase in the current ply value of our solution. Thus $s_k$ participates in the ply of $T(i', k)$. In other words, $s_k$ intersects the ply region of $Sol_{i'-1}$. By our inductive assumption, $\{Sol_{i'-1}\}$ lies relatively to the left of $\{OPT_{i'-1}\}$. Since $s_k$ covers $p_{i'}$ and $p_{i}$, therefore, $s_k$ also intersects $\{OPT_{i'-1}\}$. Since $OPT_{i'} = OPT_{i'-1}\cup \{s_k\}$, we have $[OPT_{i'}] = [OPT_{i'-1}] + 1$.  \\
So, $[T(i', k)] = [Sol_{i'-1}]+1\leq [OPT_{i'-1}]+1 = [OPT_{i'}]$. If we picked $s_k$, i.e., $T(i', k)$ for $p_{i'}$ then we are done. If we did not pick $s_k$, then suppose we picked $s_{k'}$ that does not contain $p_{i}$. The entry $T(i', k)$ will continue to be present in the table till the $i$-th row. Eventually, in the $i$-th iteration our algorithm will pick it since it does not lead to an increase in ply. Thus we have $[T(i', k)]\leq [OPT_{i}]$ in our table. Hence $[Sol_{i}] > [OPT_{ij^{*}}]$ is ruled out.\\
\end{proof}
This completes the proof of the maintenance of the loop invariant.
\end{proof}
The argument for the possibility (\ref{eqn:1b}) is similar. Now consider the running time of Algorithm (\ref{algo:exact_dp}). The sorting of points takes $O(n\log n)$ time. There are $n\cdot m$ entries in $T$ to be computed. For computing the entry $T(i, j)$, all the feasible entries (at most $m$) in row $(i-1)$ need to be considered. Given a set of squares $F$, computing the ply region of $F\cup \{s_j\}$ takes $O(1)$ time. Hence the overall running time is $O(nm^{2})$, where $n$ is the number of points and $m$ is the number of squares.
\begin{theorem}
Algorithm (\ref{algo:exact_dp}) is a polynomial time exact algorithm for the \textbf{MPCSIHL1} problem (\ref{defn:algo_mpsc_horiline}).
\end{theorem}
\begin{proof}
We have already seen that Algorithm (\ref{algo:exact_dp}) runs in polynomial time. Due to the Claim (\ref{claim:loop_invariant}), the $i$-th row in $T$ contains an optimal solution for $P_i$. Now consider the $n$-th row of the table $T$ corresponding to the point $p_{n}$. The entry achieving minimum ply is the optimum solution for our input point set $P=P_{n}$.
\end{proof}

\begin{theorem}
The minimum ply cover for a given set of points lying within the span of a given set of squares intersected by a horizontal line can be approximated within a $2$ factor in polynomial time.
\end{theorem}
\begin{proof}
We run the Algorithm (\ref{algo:exact_dp}) twice. Once for the subproblem corresponding to the points lying below the intersecting horizontal line and once for the subproblem corresponding to the points lying above the intersecting horizontal line. Then we return the union of the two solutions, say $Sol_{a}\cup Sol_{b}$. Clearly, our solution covers all the input points and hence is a feasible set cover. Now consider an arbitrary point on the plane. It is covered by at most $[Sol_{a}]+[Sol_{b}]$ squares, where $ply(Sol_{a})$ (resp. $ply(Sol_{b})$) is the ply of the problem defined for points above (resp. below) the horizontal line. Since $ply(Sol_a)\leq ply(OPT)$ and $ply(Sol_b)\leq ply(OPT)$, hence our solution is a $2$-approximation.
\end{proof}

\subsection{Points are in a unit height horizontal slab}
We consider the special case where the input points lie within a horizontal slab of height $1$ and the input squares intersect either the top boundary or the bottom boundary of the slab. Refer to Figure (\ref{fig:horizontal_slab}). For this case, we give a $(9+\epsilon)$-factor approximation algorithm for computing the minimum ply cover, where $\epsilon>0$ is a small constant. The algorithm is simple to state and analyze. We define our problem formally.
\begin{definition}\label{defn:minply_slab}
Consider two horizontal lines $L_1$ and $L_2$ unit distance apart, where $L_2$ is above $L_1$. Now consider a set $\mathcal{S}$ of unit squares where each square intersects either $L_1$ or $L_2$. Let $P$ be a set of points located below $L_2$ but above $L_1$, each point lying inside at least one of the squares in $\mathcal{S}$. The goal is to compute the minimum ply cover of $P$ with the squares in $\mathcal{S}$.
\end{definition}
\begin{theorem}\label{thm:3times}
For unit squares, if there exists a $c$-approximation for the Problem \ref{defn:minply_slab}, then there exists a $3c$-approximation for the minimum ply cover problem.
\end{theorem}
\begin{proof}
We partition the plane into horizontal slabs of unit height. If there are $n$ input points then there are at most $n$ horizontal slabs containing at least $1$ point. Suppose we denote the slabs from bottom to top as $H_{1}, H_{2}, \ldots, H_{n}$. We solve for each slab and return the union of the solutions as the final output. Our solution is a feasible solution for all the input points. Consider an arbitrary point on the plane. This point lies within some horizontal slab, say $H_{i}$. This point may lie in some max clique, i.e., ply region of $H_{i}$. Simultaneously, this point may lie within the max cliques of the solution for $H_{i-1}$ and $H_{i+1}$. Let $Sol_{i}$ denote the solution for the slab $H_{i}$ returned by our algorithm. Let $OPT_{i}$ denote the optimal solution for the slab $H_{i}$. Let $Sol = \cup_{i} Sol_{i}$. Clearly, $ply(OPT_{i})\leq ply(OPT)$ for all $i$, where $OPT$ is the minimum ply for the entire input. Hence
\begin{align*}
ply(Sol) &= \max_{i}(ply(Sol_{i-1}) + ply(Sol_{i}) + ply(Sol_{i+1})\\
&\leq c\cdot (ply(OPT_{j-1}) + ply(OPT_{j}) + ply(OPT_{j+1}))\hspace{5mm}\text{[$j:$ index at which the sum is max.]}   \\
&\leq 3c\cdot ply(OPT)
\end{align*}
\end{proof}

We use a similar greedy algorithm as in our previous section. Here the greedy tie-breaking rule is slightly modified as shown below.
\begin{itemize}
    \item Select the cover with the minimum ply value.
    \item In case of a tie, prefer a \textit{floating} ply region over an \textit{anchored} ply region.
    \item In case of a further tie, select the cover with the leftmost ply region right side.
    \item In case of a further tie, select the cover with the narrowest ply region.
    \item In case of a further tie, select a cover arbitrarily from the tied covers.
\end{itemize}

\begin{figure}[ht!]
\centering
\includegraphics[width=8cm]{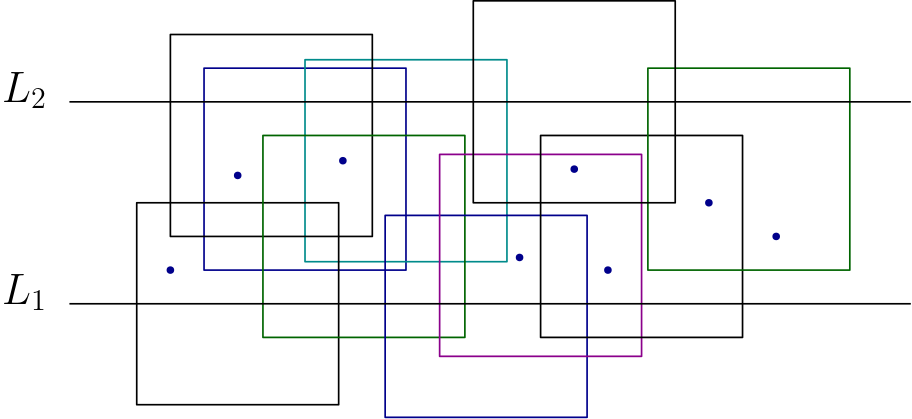}
\caption{The unit height horizontal slab subproblem.}
\label{fig:horizontal_slab}
\end{figure}
For $1\leq i\leq n$, $1\leq j\leq m$, we denote the $(i, j)$-th entry of the table as $T(i, j)$. We denote the ply of a minimum ply solution in the $i$-th row by $ply_{i}$. Recall that by $P_{i}$, we denote the $i$ leftmost input points. In this section, we will prove the following lemma.
\begin{lemma}\label{lem:3_approx}
The greedy algorithm is an $(9+\epsilon)$-approximation algorithm for the problem (\ref{defn:minply_slab}).
\end{lemma}

\begin{claim}\label{claim:bounds}
    For any $j$, if the solution $T(i+1, j)$ is feasible for $P_{i+1}$, then the ply of $T(i+1, j)$ lies between $ply_{i}$ and $ply_{i}+1$. Notationally,
    \begin{align}\label{eqn:rowwise_ply}
        ply_{i} \leq ply(T(i+1, j)) \leq ply_{i} + 1
    \end{align}
\end{claim}
\begin{proof}
    The upper bound is straightforward to prove. Since the solution $T(i+1, j)$ can be composed as the union of the $j$-th square and the minimum ply solution in the $i$-th row; hence the ply of $T(i+1, j)$ can be at most one more than $ply_{i}$.\\
    We prove the lower bound by contradiction. Suppose, for some $j$, the ply of $T(i+1, j)$ is strictly less than $ply_{i}$. So, $s_j$ does not belong to the optimum solution in the $i$-th row of the table computed by our algorithm. Two cases are possible.\\ 
    Case $1$: While processing the point $p_{i+1}$, $s_{j}$ can be combined with a minimum ply solution in the $i$-th row, say, $T(i, k)$ such that $T(i, k)\cup \{s_{j}\}$ is a feasible solution for $P_{i+1}$ and $s_{j}$ does not intersect the ply region of $T(i, k)$ and at least one square participating in the ply region of $T(i, k)$ is discarded. Thus $ply(T(i+1, j))<ply_{i}$. \\
    Case $2$: While processing the point $p_{i+1}$, $s_{j}$ can be combined with a minimum ply solution in the $i$-th row, say, $T(i, k)$ such that $T(i, k)\cup \{s_{j}\}$ is a feasible solution for $P_{i+1}$ and $s_{j}$ intersects the ply region of $T(i, k)$ and at least two squares participating in the ply region of $T(i, k)$ are discarded.  Thus $ply(T(i+1, j))<ply_{i}$. \\
    The ply region of any solution is bound to be a subset of the ply region of its parent solution in the previous row. Hence, in both the cases above, if $s_{j}$ were a better pick in the $(i+1)$-th row it would have been picked earlier by our greedy algorithm. Hence, we have arrived at a contradiction.
\end{proof}
If a set of squares have a common intersection, we say that they form a clique. We call the common intersection region of the clique as the \textit{ply region}. First, we will classify the cliques in a solution into some distinct types. Let the size of the clique under consideration be $l$, i.e., $l$ squares have a common intersection. If the top side of a square $s_1$ lies above the top side of another square $s_2$, we say that $s_1$ lies above $s_2$. Equivalently, we can say that $s_2$ lies below $s_1$. A set of squares $s_1, s_2, \ldots, s_k$ such that $s_i$ is above $s_{i+1}$ for all $1\leq i\leq k-1$, is termed as a set of \textit{descending} squares. A set of squares $s_1, s_2, \ldots, s_k$ such that $s_i$ is below $s_{i+1}$ for all $1\leq i\leq k-1$, is termed as a set of \textit{ascending} squares. If the left side of a square $s_1$ lies to the left of the left side of another square $s_2$, we say that $s_1$ lies to the left of $s_2$. We denote this by $s_1\prec s_2$. The following types of clique are possible.
\begin{itemize}
    \item \textbf{Top-anchored}: Here all the constituent squares of the clique intersect the top line $L_2$. There are three subtypes as shown in the Figure (\ref{fig:top_all}).
    \begin{itemize}
        \item \textbf{Top-anchored ASC}: Here the squares from left to right are ascending. In other words, for any two squares $s_1, s_2$ in the clique such that $s_1\prec s_2$, the square $s_2$ is above $s_1$.
        \item \textbf{Top-anchored DESC}: Here the squares from left to right are descending. In other words, for any two squares $s_1, s_2$ in the clique such that $s_1\prec s_2$, we have $s_2$ below $s_1$.
        \item \textbf{Top-anchored DESC$|$ASC}: There is a $k\geq 2$ such that the squares constituting the clique are initially descending from $s_1$ to $s_k$. Then the square $s_{k+1}$ lies above $s_k$. Then the squares $s_{k+1}$ to $s_{l}$ are ascending. We term the square $s_{k+1}$ as the \textit{transition square}. This type of clique can be viewed as the merger of a descending clique with an ascending clique, in that order.
    \end{itemize}
\end{itemize}

\begin{figure}[ht!]
\centering
\includegraphics[width=10cm]{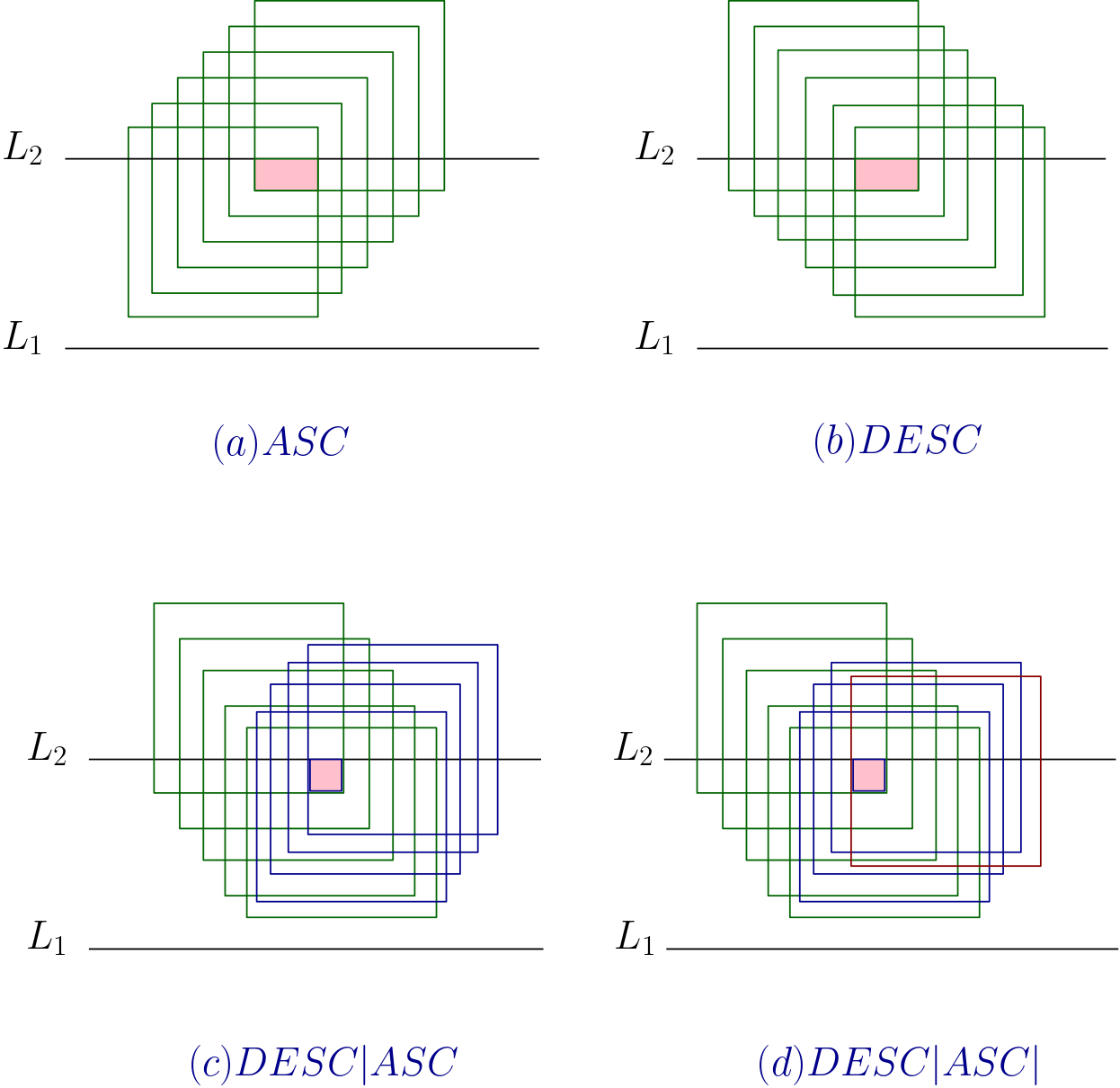}
\caption{(a), (b) and (c) shows three different types of top-anchored cliques. (d) Shows an invalid clique where an DESC$|$ASC top-anchored clique is followed by a transition square.}
\label{fig:top_all}
\end{figure}

\begin{claim}\label{claim:forbidden_1}
    A top anchored clique of type ASC$|$ASC is forbidden.
\end{claim}
\begin{proof}
    Suppose not. Suppose, the left ascending sequence consists of $k$ squares. Then the rightmost square $s_k$ of the first ASC sequence will become redundant since the second last square $s_{k-1}$ of the first ASC sequence and the transition square $s_{k+1}$ will fully cover the relevant area of $s_{k}$. This is a contradiction since $s_k$ is redundant. Refer to Figure (\ref{fig:top_asc_asc}) for an example. 
\end{proof}

\begin{figure}[ht!]
\centering
\includegraphics[width=5cm]{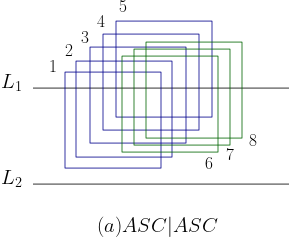}
\caption{A forbidden clique of type top-anchored ASC$|$ASC. Here the square $5$ is redundant since $4$ and $6$ fully cover the relevant area of $5$.}
\label{fig:top_asc_asc}
\end{figure}

\begin{claim}\label{claim:forbidden_2}
    A top anchored clique of type DESC$|$DESC is forbidden.
\end{claim}
\begin{proof}
Suppose not. Suppose, the left descending sequence consists of $k$ squares. Then the leftmost square $s_{k+1}$ of the second DESC sequence will become redundant since the ast square $s_{k}$ of the first DESC sequence and the square $s_{k+2}$ will fully cover the relevant area of $s_{k+1}$. This is a contradiction since $s_{k+1}$ is redundant. Refer to Figure \ref{fig:top_aaa_desc}(a) for an example.
\end{proof}

\begin{figure}[ht!]
\centering
\includegraphics[width=10cm]{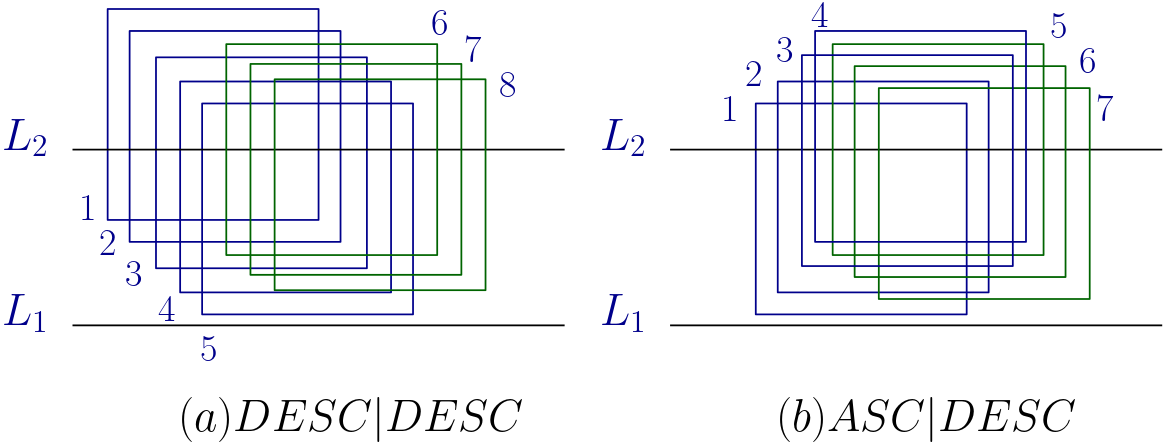}
\caption{(a) A forbidden clique of type top-anchored DESC$|$DESC. Here the square $6$ is redundant since the squares $5$ and $7$ fully cover the relevant area of $6$. (b) A forbidden clique of type top-anchored ASC$|$DESC. Here the square $4$ is redundant since the squares $3$ and $5$ fully cover the relevant area of $4$.}
\label{fig:top_aaa_desc}
\end{figure}
\begin{claim}
    A top anchored clique of type ASC$|$DESC is forbidden.
\end{claim}
\begin{proof}
The proof is similar to the proof of Claim (\ref{claim:forbidden_1}). Refer to Figure \ref{fig:top_aaa_desc}(b) for an example.
\end{proof}

\begin{figure}[ht!]
\centering
\includegraphics[width=10cm]{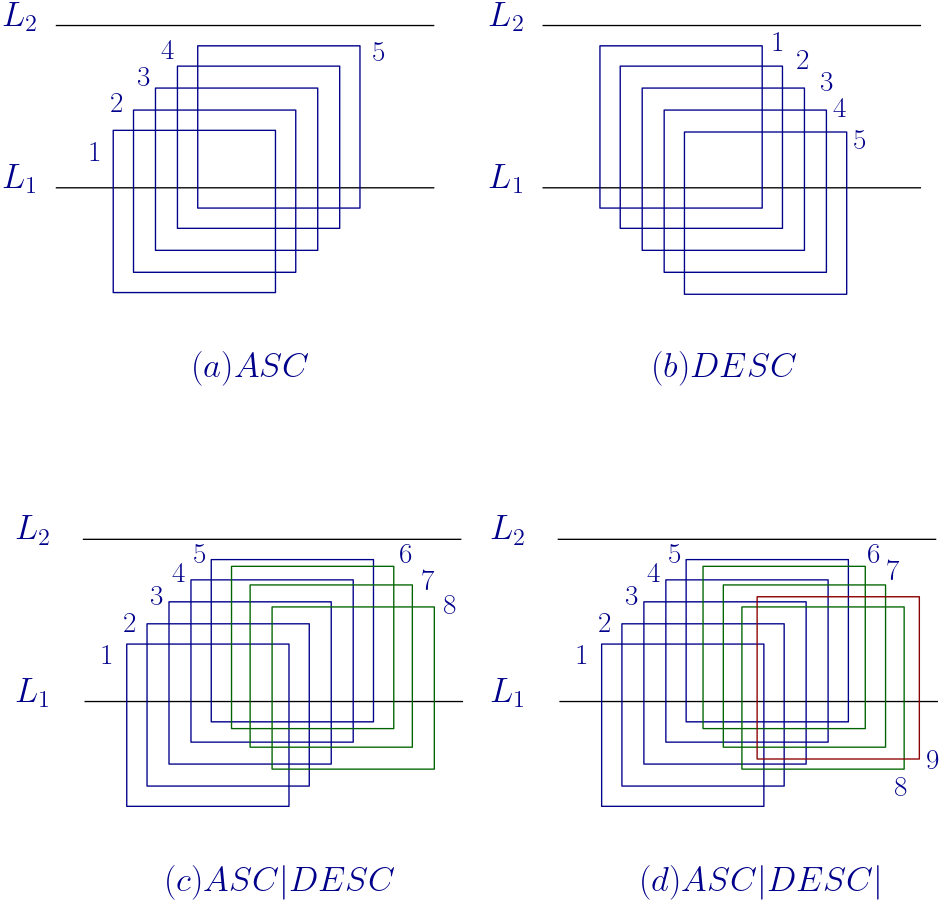}
\caption{(a), (b) and (c) shows three different types of bottom-anchored cliques. (d) Shows an invalid clique where an ASC$|$DESC bottom-anchored clique is followed by a transition square.}
\label{fig:bottom_all}
\end{figure}

\begin{itemize}
    \item \textbf{Bottom-anchored}: Here all the constituent squares of the clique intersect the bottom line $L_1$. There are three subtypes as shown in the Figure (\ref{fig:bottom_all}).
    \begin{itemize}
         \item \textbf{Bottom-anchored ASC}: Here the squares from left to right are ascending. In other words, for any two squares $s_1, s_2$ in the clique such that $s_1\prec s_2$, the square $s_2$ is above $s_1$.
         \item \textbf{Bottom anchored DESC}: Here the squares from left to right are descending. In other words, for any two squares $s_1, s_2$ in the clique such that $s_1\prec s_2$, the square $s_2$ is below $s_1$.
        \item \textbf{Bottom anchored ASC$|$DESC}: There is a $k\geq 2$ such that the squares constituting the clique are initially ascending from $s_1$ to $s_k$. Then the square $s_{k+1}$ lies below $s_k$. Then the squares $s_{k+1}$ to $s_{l}$ are descending.
    \end{itemize}
\end{itemize}

\begin{claim}\label{claim:bot_forbid}
    A bottom anchored clique of type ASC$|$ASC, DESC$|$ASC or DESC$|$DESC is forbidden.
\end{claim}
\begin{proof}
    In all the three cases some squares will become redundant leading to a contradiction. The proof is similar to the proof of Claim (\ref{claim:forbidden_1}). Refer to Figures (\ref{fig:bottom_asc_asc}) and (\ref{fig:bottom_desc_aaa}) for examples.
\end{proof}

\begin{figure}[ht!]
\centering
\includegraphics[width=5cm]{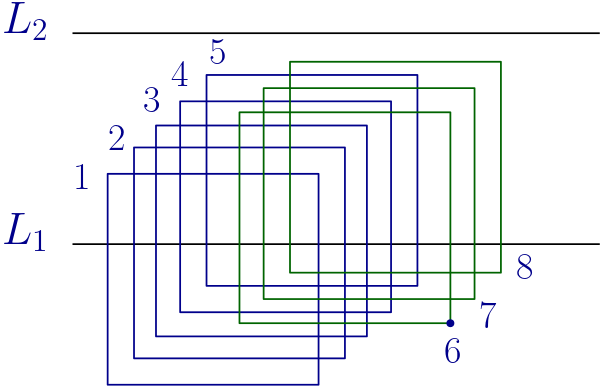}
\caption{A forbidden clique of type bottom-anchored ASC$|$ASC. Here the square $6$ is redundant since the squares $5$ and $7$ fully cover the relevant area of $6$.}
\label{fig:bottom_asc_asc}
\end{figure}

\begin{figure}[ht!]
\centering
\includegraphics[width=9cm]{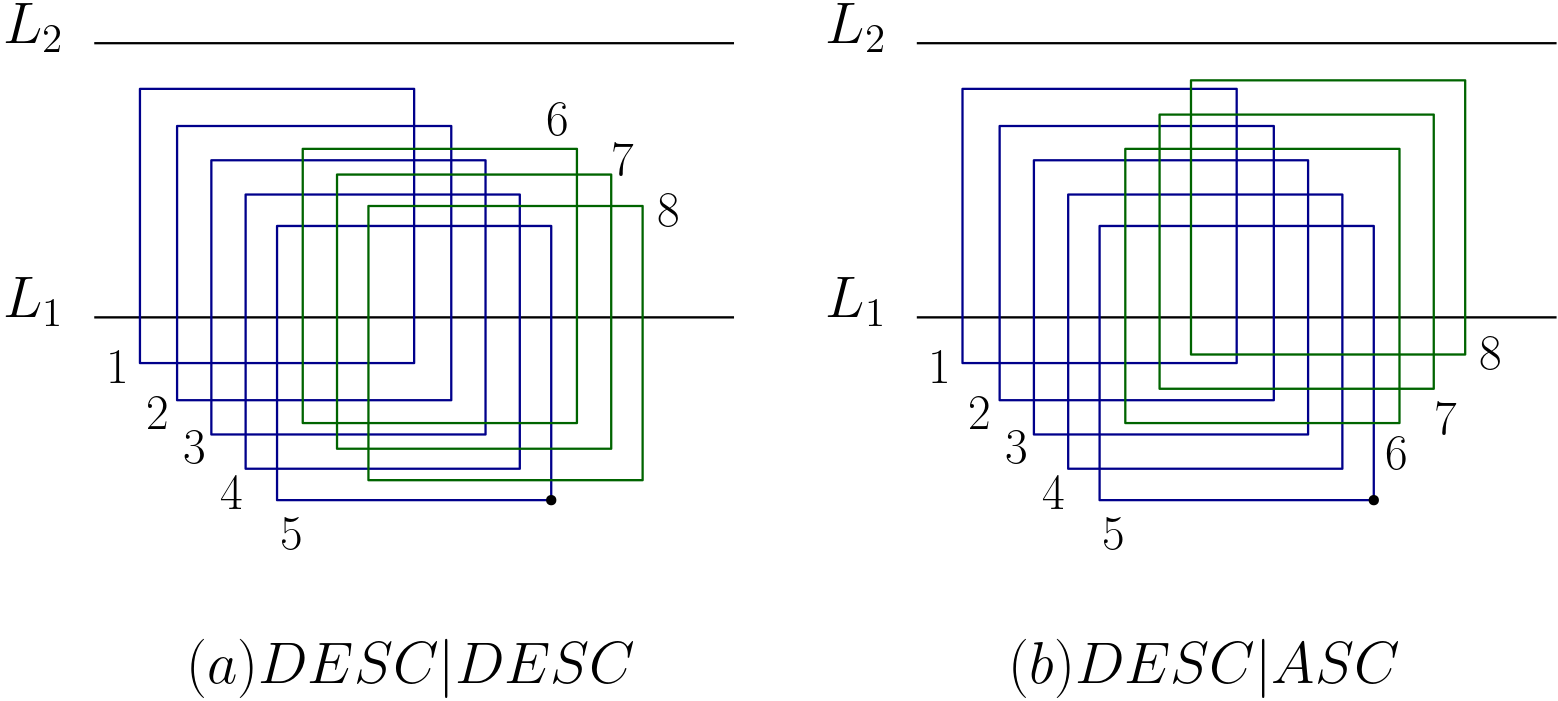}
\caption{(a) A forbidden clique of type bottom-anchored DESC$|$DESC. Here the square $5$ is redundant since the squares $4$ and $6$ fully cover the relevant area of $5$. (b) A forbidden clique of type bottom-anchored DESC$|$ASC. Here the square $5$ is redundant since the squares $4$ and $6$ fully cover the relevant area of $5$.}
\label{fig:bottom_desc_aaa}
\end{figure}

\begin{itemize}
    \item \textbf{Floating}: If a set of squares have a common intersection and some of the squares intersect the top line $L_2$ while others intersect the bottom line $L_1$, we call the common intersection as a \textit{floating} clique. In the subtypes below, at least one of the squares intersects the bottom line $L_1$ and at least one of the squares intersects the top line $L_2$.
    \begin{itemize}
        \item \textbf{Floating ASC}: Here, the leftmost square $s_1$ intersects the bottom line $L_1$. The squares from left to right are ascending. In other words, for any two squares $s_1, s_2$ in the clique such that $s_1\prec s_2$, the square $s_2$ is above $s_1$. The rightmost square $s_{l}$ must intersect the top line $L_2$. Refer to Figure (\ref{fig:floating_monotonic}(a)).
        \item \textbf{Floating DESC}: Here, the leftmost square $s_1$ intersects the top line $L_2$. the squares from left to right are descending. In other words, for any two squares $s_1, s_2$ in the clique such that $s_1\prec s_2$, the square $s_2$ is below $s_1$.  The rightmost square $s_{l}$ must intersect the bottom line $L_1$. Refer to Figure (\ref{fig:floating_monotonic}(b)).
        
\begin{figure}[ht!]
\centering
\includegraphics[width=10cm]{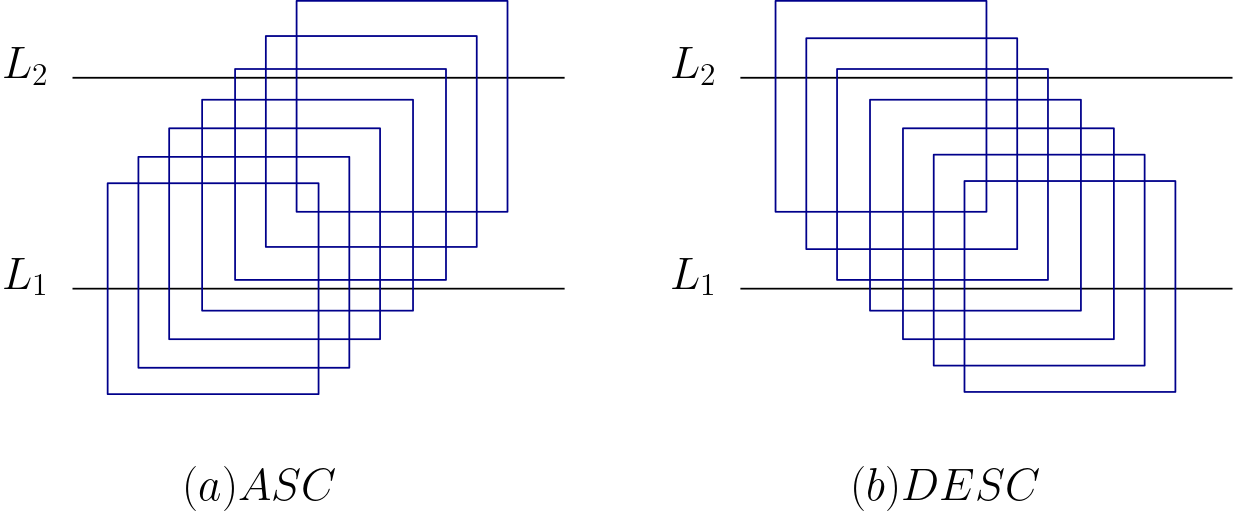}
\caption{Types of monotonic floating cliques.}
\label{fig:floating_monotonic}
\end{figure}

        \item \textbf{Floating ASC$|$ASC}: There is a $k\geq 2$ such that the squares constituting the clique are initially ascending from $s_1$ to $s_k$. Then the square $s_{k+1}$ lies below $s_k$. Then the squares $s_{k+1}$ to $s_{l}$ are again ascending. Refer to Figure (\ref{fig:floating_asc_asc}). 
        A clique of this type can be thought of as the merger of two monotonic ascending cliques, where the first clique is composed of the squares $s_1$ through $s_k$ and the second clique is composed of the squares $s_{k+1}$ through $s_l$. The structure of such a clique follows certain rules as specified by the following claim.
        \begin{claim}\label{claim:float_forbid_asc_asc}
            In a floating clique of type ASC$|$ASC, at most one square of the first ascending sequence can intersect the top line $L_2$ and, at most one square of the second ascending sequence can intersect the bottom line $L_1$.
        \end{claim}
\begin{proof}
    Suppose not. There are at least two squares 
 in the first ascending sequence intersecting $L_2$. Then the two rightmost squares in the first ascending sequence $s_{k-1}$ and $s_{k}$ definitely intersect $L_2$. By definition, $s_{k+1}$ lies below $s_{k}$. If $s_{k+1}$ intersects the top line $L_{2}$, then $s_{k}$ is redundant as $s_{k-1}$ and $s_{k+1}$ cover the relevant area of $s_k$. Refer to Figure \ref{fig:floating_asc_asc}(a). If $s_{k+1}$ intersects the bottom line $L_{1}$, then there are two cases. Case $1$: $s_{k+2}$ also intersects the bottom line $L_{1}$, then $s_{k+1}$ is redundant as $s_1, s_{k}$ and $s_{k+2}$ cover the relevant area of $s_{k+1}$. Refer to Figure \ref{fig:floating_asc_asc}(b). Case $2$: $s_{k+2}$ intersects the top line $L_{2}$, then $s_{k}$ is redundant as $s_{k-1}, s_{k+1}$ and $s_{k+2}$ cover the relevant area of $s_{k}$. Refer to Figure \ref{fig:floating_asc_asc}(c).

 Now consider the second part of the claim. Suppose there are at least two squares in the second ascending sequence intersecting $L_1$. Then its two leftmost squares $s_{k+1}$ and $s_{k+2}$ definitely intersects $L_1$. The square $s_{k+1}$ is redundant as $s_{1}, s_{k}$ and $s_{k+2}$ cover the relevant area of $s_{k+1}$. Refer to Figure \ref{fig:floating_asc_asc}(d).
\end{proof}

\begin{figure}[ht!]
\centering
\includegraphics[width=10cm]{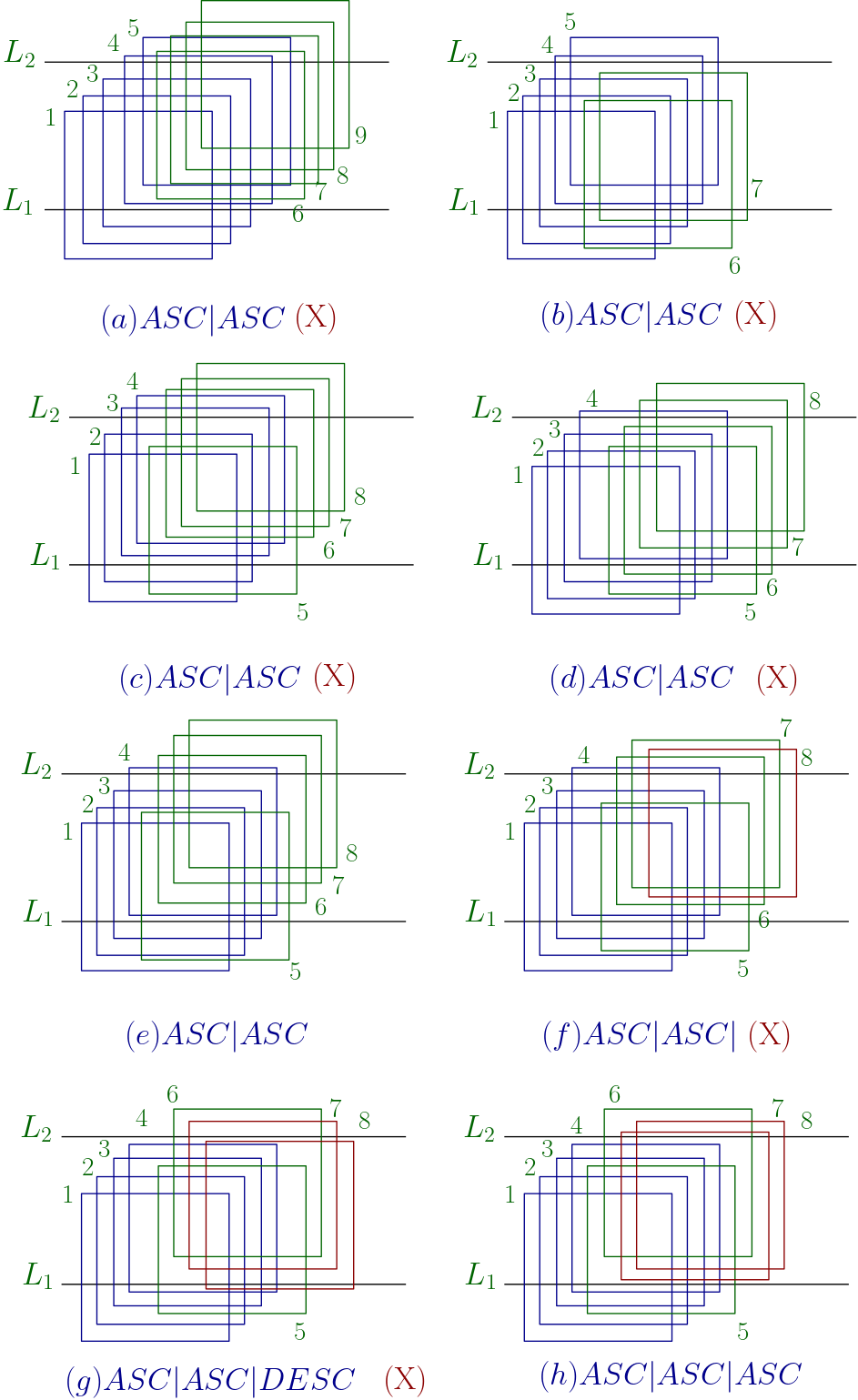}
\caption{Floating ASC$|$ASC cliques. In (a), (b) and (c), the first ASC sequence has more than $1$ square intersecting $L_2$. (a) Shows an invalid clique of ASC$|$ASC type, where the rightmost square of the first ASC sequence, $5$ is covered by $4$ and $6$. The second ASC sequence has more than $1$ squares intersecting $L_1$. (b) Shows an invalid clique of ASC$|$ASC type where, the transition square $6$ is covered by $1, 5$ and $7$. (c) Shows an invalid clique of ASC$|$ASC type where, the rightmost square of the first ASC sequence, $4$ is covered by $3, 5$ and $6$. (d) Shows an invalid clique where an ASC$|$ASC the second ascending sequence has more than $1$ square intersecting $L_1$. Here, the relevant area of $5$ is covered by $1, 4$ and $6$. (e) Shows a valid clique of type ASC$|$ASC. (h) Shows an invalid clique of type ASC$|$ASC followed by a transition square. (g) Shows an invalid clique of type ASC$|$ASC$|$DESC. (h) Shows a valid clique of type ASC$|$ASC$|$ASC.}
\label{fig:floating_asc_asc}
\end{figure}

        \item \textbf{Floating ASC$|$DESC}: There is a $k\geq 2$ such that the squares constituting the clique are initially ascending from $s_1$ to $s_k$. Then the square $s_{k+1}$ lies below $s_k$. Then the squares $s_{k+1}$ to $s_{l}$ are descending. Refer to Figure (\ref{fig:floating_asc_desc}).
        A clique of this type can be thought of as the merger of a monotonic ascending clique followed by a monotonic descending clique, where the first clique is composed of the squares $s_1$ through $s_k$ and the second clique is composed of the squares $s_{k+1}$ through $s_l$. The structure of such a clique follows certain rules as specified by the following claim.
        \begin{claim}
            In a clique of type ASC$|$DESC, at most two squares of the clique can intersect the top line $L_2$.
        \end{claim}
    \begin{proof}
    Suppose not. There are at least three squares intersecting $L_2$. The square $s_k$ is the topmost square in the clique, hence $s_k$ definitely intersects $L_2$. There are three cases. \\
    Case $1$: $s_{k-2}, s_{k-1}, s_k$ intersect $L_2$. The square $s_{k-1}$ is redundant as $s_{k-2}, s_{k}$ and $s_{k+1}$ cover the relevant area of $s_{k-1}$. Refer to Figure \ref{fig:floating_asc_desc}(a).\\
    Case $2$: $s_{k-1}, s_k, s_{k+1}$ intersect $L_2$. The square $s_{k}$ is redundant as $s_{k-1}$ and $s_{k+1}$ cover the relevant area of $s_{k}$. Refer to Figure \ref{fig:floating_asc_desc}(b).\\
    Case $3$: $s_k, s_{k+1}, s_{k+2}$ intersect $L_2$.  The square $s_{k+1}$ is redundant as $s_{1}, s_{k}$ and $s_{k+2}$ cover the relevant area of $s_{k+1}$. Refer to Figure \ref{fig:floating_asc_desc}(c).
    
\begin{figure}[ht!]
\centering
\includegraphics[width=10cm]{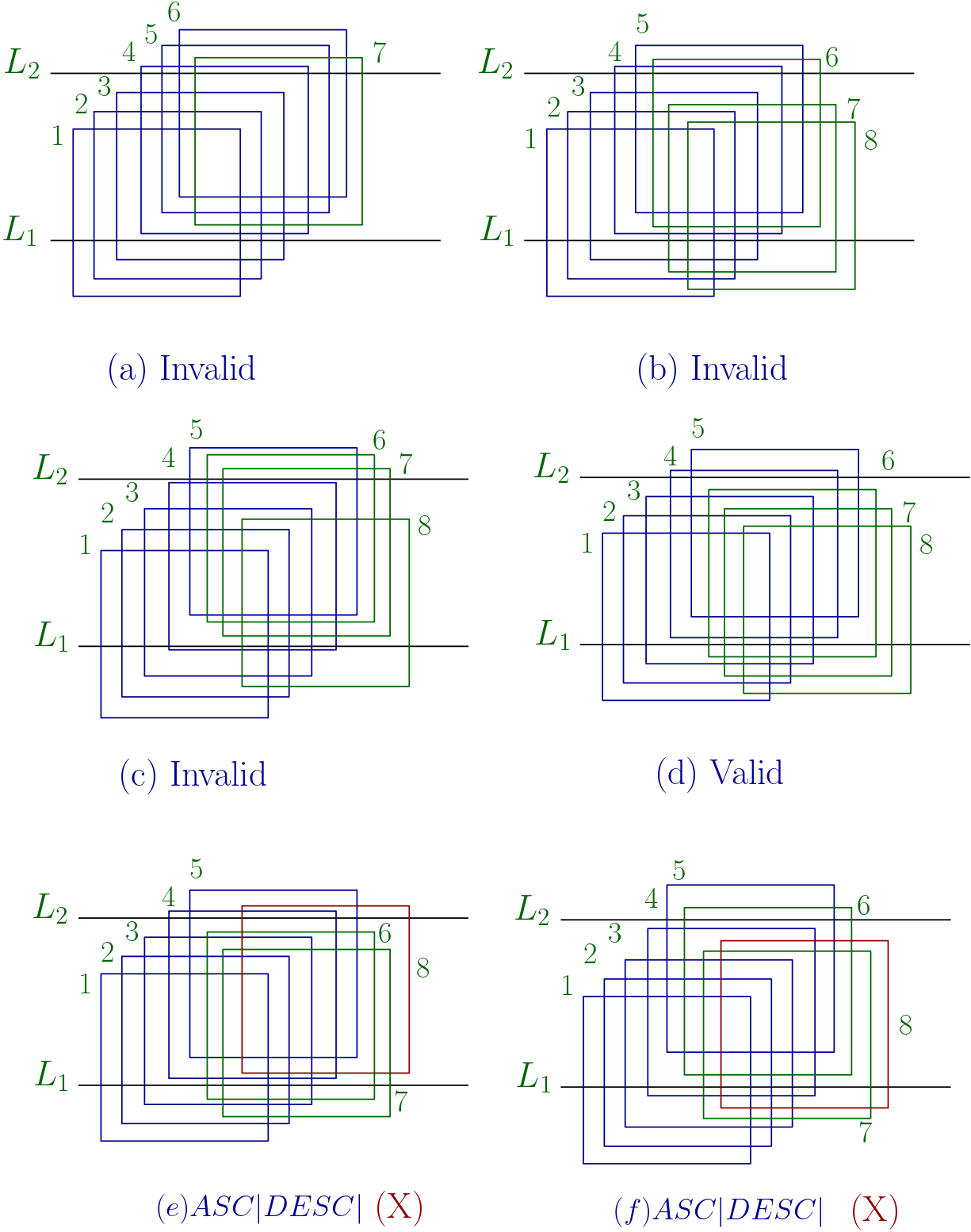}
\caption{(a), (b) and (c) show invalid cliques where more than $2$ squares intersect the top line. (a) Here the relevant area of $5$ is covered by $4, 6$ and $7$. (b) Here the relevant area of $5$ is covered by $4$ and $6$. (c) Here the relevant area of $6$ is covered by $1, 5$ and $7$. (d) A valid floating ASC$|$DESC clique where $2$ squares intersect the top line. (e) Shows an invalid clique where an ASC$|$DESC clique is followed by a transition square. The transition square $8$ intersects $L_2$. Here the relevant area of $6$ is covered by $1, 5, 7$ and $8$. (f) Shows an invalid clique where an ASC$|$DESC clique is followed by a transition square. The transition square $8$ intersects $L_1$. Here the relevant area of $7$ is covered by $1, 6$ and $8$.}
\label{fig:floating_asc_desc}
\end{figure}
Thus we have derived a contradiction in each of the cases. Hence proved.
\end{proof}
        \item \textbf{Floating DESC$|$ASC}: There is a $k\geq 2$ such that the squares constituting the clique are initially descending from $s_1$ to $s_k$. Then the square $s_{k+1}$ lies above $s_k$. Then the squares $s_{k+1}$ to $s_{l}$ are ascending. Refer to Figure (\ref{fig:floating_desc_asc}). A clique of this type can be thought of as the merger of a monotonic descending clique with another monotonic ascending clique, where the first clique is composed of the squares $s_1$ through $s_k$ and the second clique is composed of the squares $s_{k+1}$ through $s_l$. The structure of such a clique follows certain rules as specified by the following claim.
    \begin{claim}\label{claim:float_forbid_desc_asc}
        In a floating clique of type DESC$|$ASC, at most two squares of the clique can intersect the bottom line $L_1$.
    \end{claim}
    \begin{proof}
    Suppose not. There are at least three squares intersecting $L_1$. The square $s_k$ is the bottom-most square in the clique, hence $s_k$ definitely intersects $L_1$. There are three cases. \\
     Case $1$: $s_{k-2}, s_{k-1}, s_k$ intersect $L_1$. $s_{k+1}$ cannot intersect $L_1$ otherwise $s_{k}$ will be rendered redundant. So, $s_{k+1}$ must intersect $L_2$. Now, the square $s_{k-1}$ is redundant as $s_{k-2}, s_{k}$ and $s_{k+1}$ cover the relevant area of $s_{k-1}$. Refer to Figure \ref{fig:floating_desc_asc}(a).\\
    Case $2$: $s_{k-1}, s_{k}, s_{k+1}$ intersect $L_1$. The square $s_{k}$ is redundant as $s_{k-1}$ and $s_{k+1}$ cover the relevant area of $s_{k}$. Refer to Figure \ref{fig:floating_desc_asc}(b).\\
    Case $3$: $s_k, s_{k+1}, s_{k+2}$ intersect $L_1$.  The square $s_{k+1}$ is redundant as $s_{1}, s_{k}$ and $s_{k+2}$ cover the relevant area of $s_{k+1}$. Refer to Figure \ref{fig:floating_desc_asc}(c).\\
    Thus we have derived a contradiction in each of the cases. Hence proved.
   
    \end{proof}
        \item \textbf{Floating DESC$|$DESC}:  There is a $k\geq 2$ such that the squares constituting the clique are initially descending from $s_1$ to $s_k$. Then the square $s_{k+1}$ lies above $s_k$. Then the squares $s_{k+1}$ to $s_{l}$ are again descending. Refer to Figure (\ref{fig:floating_desc_desc}). 
        A clique of this type can be thought of as the merger of two monotonic descending cliques, where the first clique is composed of the squares $s_1$ through $s_k$ and the second clique is composed of the squares $s_{k+1}$ through $s_l$. The structure of such a clique follows certain rules as specified by the following claim.
    \begin{claim}
        In a floating clique of type DESC$|$DESC, at most one square of the first descending sequence can intersect the bottom line $L_1$ and, at most one square of the second descending sequence can intersect the top line $L_2$.
    \end{claim}
    \begin{proof}
    Suppose not. There are at least two squares 
 in the first descending sequence intersecting $L_1$. Then the two rightmost squares in the first ascending sequence $s_{k-1}$ and $s_{k}$ definitely intersect $L_1$. By definition, $s_{k+1}$ lies above $s_{k}$. Since both $s_{k-1}, s_k$ are intersecting $L_{1}$, hence $s_1$ must intersect $L_2$, otherwise $s_{k-1}$ will become redundant. If $s_{k+1}$ intersects the bottom line $L_{1}$, then $s_{k}$ is redundant as $s_{k-1}$ and $s_{k+1}$ cover the relevant area of $s_k$. Refer to Figure \ref{fig:floating_desc_desc}(a). On the other hand, if $s_{k+1}$ intersects the top line $L_{2}$, then there are two cases. \\
 Case $1$: $s_{k+2}$ also intersects the top line $L_{2}$. Then $s_{k+1}$ is redundant as $s_1, s_{k}$ and $s_{k+2}$ cover the relevant area of $s_{k+1}$. Refer to Figure \ref{fig:floating_desc_desc}(b). \\
 Case $2$: $s_{k+2}$ intersects the bottom line $L_{1}$, then $s_{k}$ is redundant as $s_{k-1}, s_{k+1}$ and $s_{k+2}$ cover the relevant area of $s_{k}$. Refer to Figure \ref{fig:floating_desc_desc}(c).

 Now consider the second part of the claim. Suppose there are at least two squares in the second descending sequence intersecting $L_2$. Then its two leftmost squares $s_{k+1}$ and $s_{k+2}$ definitely intersects $L_2$. The square $s_{k+1}$ is redundant as $s_{1}, s_{k}$ and $s_{k+2}$ cover the relevant area of $s_{k+1}$. Refer to Figure \ref{fig:floating_desc_desc}(d).\\
 Thus we have derived a contradiction in each of the cases. Hence proved.
\end{proof}
    \end{itemize}
\end{itemize}

\begin{figure}[ht!]
\centering
\includegraphics[width=10cm]{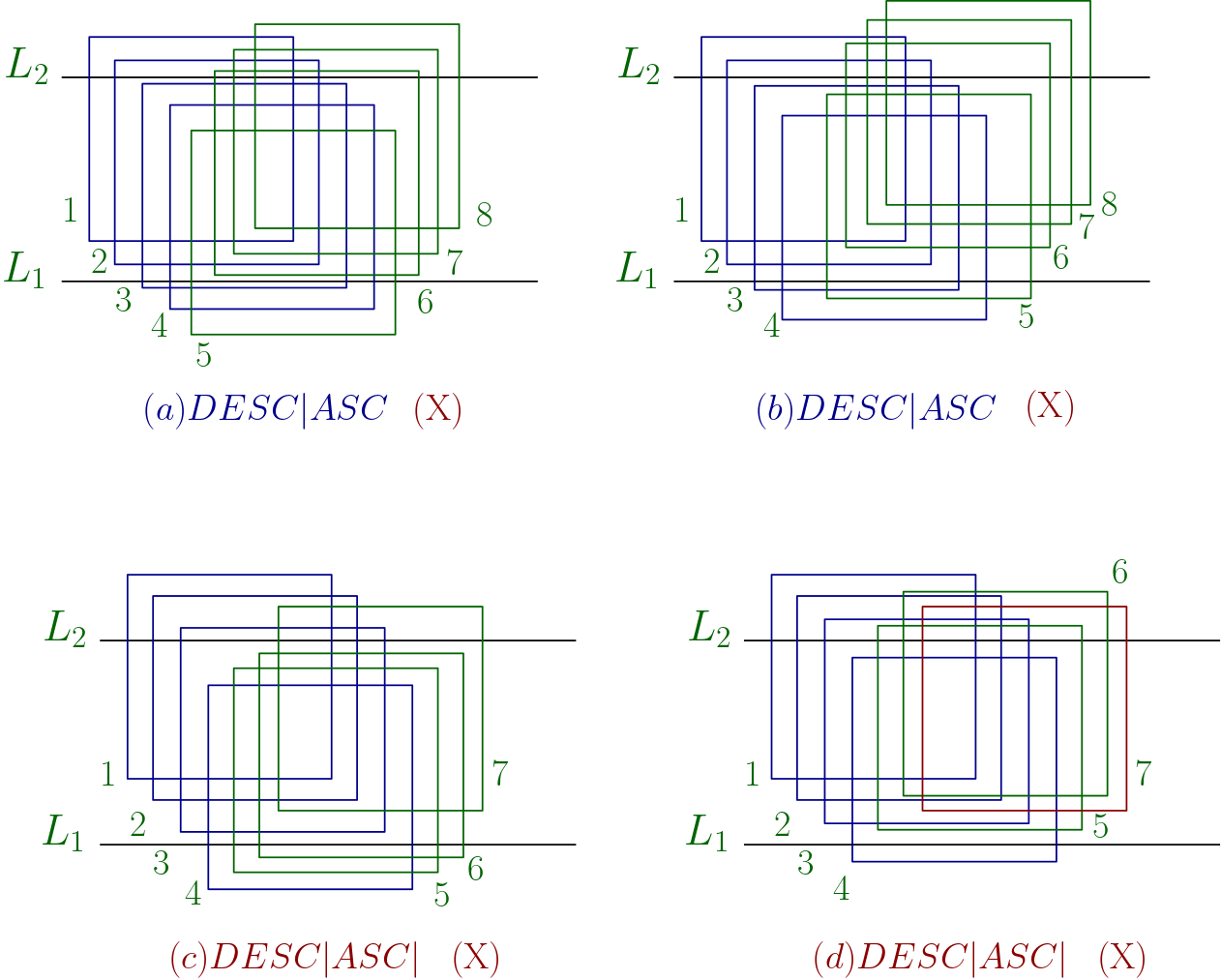}
\caption{(a), (b) and (c) show invalid DESC$|$ASC cliques where more than $2$ squares intersect the bottom line. (a) Here the relevant area of $4$ is covered by $3, 5$ and $6$. (b) Here the relevant area of $4$ is covered by $3$ and $5$. (c) Here the relevant area of $5$ is covered by $1, 4$ and $6$. (d) Shows an invalid clique where a DESC$|$ASC clique is followed by a transition square. The transition square $7$ will render either square $5$ or square $6$ redundant.}
\label{fig:floating_desc_asc}
\end{figure}

\begin{figure}[ht!]
\centering
\includegraphics[width=10cm]{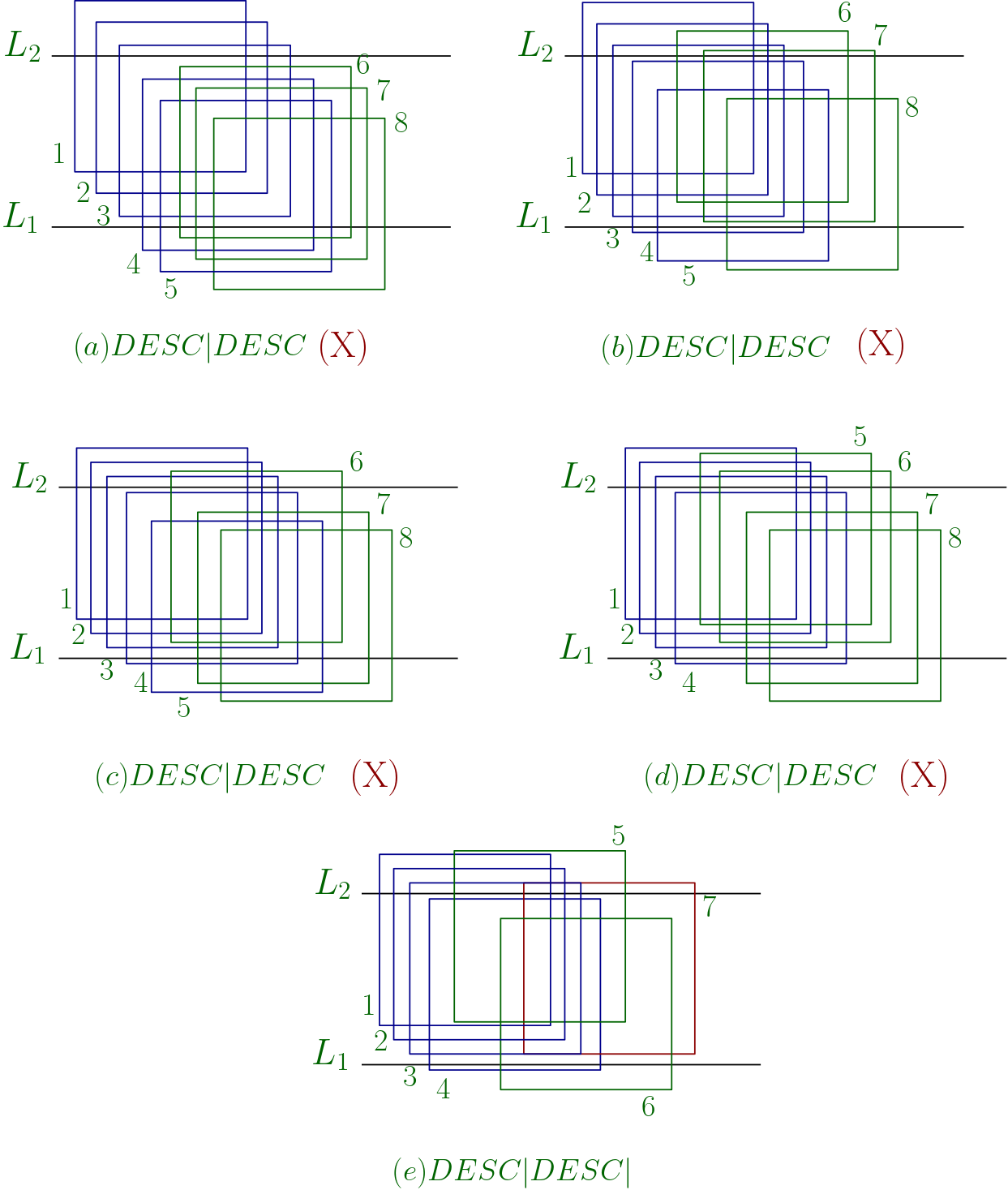}
\caption{(a), (b) and (c) show invalid DESC$|$DESC cliques where more than $1$ square from the first DESC sequence intersect the bottom line $L_1$. (a) Here the relevant area of $5$ is covered by $4$ and $6$. (b) Here the relevant area of $6$ is covered by $1, 5$ and $7$. (c) Here the relevant area of $5$ is covered by $4, 6$ and $7$. (d) Shows an invalid DESC$|$DESC clique where more than $1$ square from the second DESC sequence intersect the top line $L_2$. Here the relevant area of $5$ is covered by $1, 4$ and $6$. (d) Shows an invalid clique where a DESC$|$DESC clique is followed by a transition square. The transition square $7$ will render either square $5$ or square $6$ redundant.}
\label{fig:floating_desc_desc}
\end{figure}

\begin{lemma}
If three monotonic sequences of squares $S_1, S_2, S_3$, from left to right respectively, merge to form a clique then $S_2$ consists of at most $2$ squares.
\end{lemma}
\begin{proof}
The sequence of squares $S_2$ is either ASC or DESC. We consider the ASC case first. Suppose for the sake of contradiction that there are at least $3$ squares in the sequence $S_2$. We denote the three leftmost ones from left to right as $s_{k+1}, s_{k+2}$ and $s_{k+3}$. We know from the claims (\ref{claim:bot_forbid}), (\ref{claim:float_forbid_asc_asc}) and (\ref{claim:float_forbid_desc_asc}) that $S_2$ cannot be a bottom-anchored clique. We have the following possibilities,

i) $S_2$ is top-anchored ASC: Suppose, $|S_1| = k$, $|S_1|+|S_2| = l$. If the transition square of $(S_2, S_3)$, i.e., $s_{l+1}$ is top-intersecting then, $s_l$ will be covered by $s_{l-1}$ and $s_{l+1}$. Thus $s_{l}$ will become redundant. A contradiction. If $s_{l+1}$ is bottom-intersecting then, the area of $s_{l-1}$ will be covered $s_{k+1}, s_l$ and $s_{l+1}$. Thus $s_{l-1}$ will become redundant. A contradiction.

ii) $S_2$ is floating ASC: If $S_1$ is bottom-anchored DESC then the rightmost square of $S_1$, i.e., $s_{k}$ will become redundant as its relevant area will be covered by $s_{k-1}$ and $s_{k+1}$. Hence, $S_1$ is either Floating or Bottom-anchored ASC. Now the following cases are possible.

(a) If $s_{k+2}$, $s_{k+3}$ and the transition square of $(S_2, S_3)$, i.e., $s_{l+1}$ are top-intersecting then the relevant area of $s_{k+3}$ will be covered by $s_{k+2}$ and $s_{l+1}$. Thus $s_{k+3}$ will become redundant.\\
(b) If $s_{k+2}, s_{k+3}$ are top-intersecting but the transition square of $(S_2, S_3)$, i.e., $s_{l+1}$ is bottom-intersecting then there are two cases. Case $1$: the transition square of $(S_1, S_2)$, i.e., $s_{k+1}$ is bottom-intersecting, then the relevant area of $s_{k+1}$ will be entirely covered by $s_1$, $s_k, s_{k+2}$ and $s_{l+1}$. Thus $s_1$ will become redundant. Case $2$: the transition square between $S_1, S_2$, i.e., $s_{k+1}$ is top-intersecting, then the area of $s_{k+2}$ will be entirely covered by $s_{k+1}$, $s_{k+3}$ and $s_{l+1}$. Thus $s_{k+2}$ will become redundant.\\
(c) If $s_{k+1}$ and $s_{k+2}$ are bottom-intersecting, $s_{k+3}$ is top-intersecting and the transition square of $(S_2, S_3)$, i.e., $s_{l+1}$) is top-intersecting then there are two cases. \\
Case $1$: The rightmost square of $S_1$, i.e., $s_{k}$ is bottom-intersecting, then the relevant area of $s_{k+1}$ will be entirely covered by $s_{k}$ and $s_{k+2}$. Thus $s_{k+1}$ will become redundant. \\
Case $2$: The rightmost square of $S_1$, i.e., $s_{k}$ is bottom-intersecting, then the relevant area of $s_{k+3}$ will be entirely covered by $s_k, s_{k+2}$ and $s_{l+1}$. Thus $s_{k+3}$ will become redundant.\\
(d) If $s_{k+1}$ and $s_{k+2}$ are bottom-intersecting, $s_{k+3}$ is top-intersecting and the transition square of $(S_2, S_3)$, i.e., $l_{l+1}$ is bottom-intersecting then there are two cases. \\
Case $1$: the transition square between $S_1, S_2$, i.e., $s_{k+1}$ is bottom-intersecting, then the area of $s_{k+2}$ will be entirely covered by $s_{k+1}$ and $s_{k+3}$. Thus $s_{k+2}$ will become redundant. \\
Case $2$: the transition square between $S_1, S_2$, i.e., $s_{k+1}$ is top-intersecting, then the area of $s_{k+3}$ will be entirely covered by $s_{k+1}, s_{k+2}, s_{k+4}$, and $s_{k+5}$. Thus $s_{k+3}$ will become redundant.

We have shown a contradiction for each of the possibilities when $S_2$ is a sequence of ascending type. Similar arguments are applicable when $S_2$ is a sequence of descending squares. 
\end{proof}

\begin{claim}\label{claim:placement}
    Our algorithm irrevocably chooses a square $s$ at a point at or to the left of the rightmost exclusive point of $s$.
\end{claim}
\begin{proof}
Consider a square $s\in Sol$. The square $s$ must have been included into $Sol$ during processing some point $p\in s$. Let $p_{r}$ be the rightmost exclusive point of $s$ in $Sol$. Suppose $s$ does not exist in the partial solution obtained for the points till the point $p_{r}$. Then $s$ must have been included at some point $p_q$ to the right of $p_r$. While processing the point $p_q$, the algorithm must have discarded some square(s) so that $p_{r}$ and other points in $Excl(s)$ can become exclusive to $s$. The resulting solution is a feasible solution for $P_{q}$. This means that if $s$ was a better pick for $p_{q}$, it would have been picked earlier by our greedy algorithm. Hence, we have arrived at a contradiction.
\end{proof}
\begin{claim}
    When a clique is considered separately, the exclusive regions of all non-extreme squares in the clique are rectangular or $L$-shaped. All except at most two non-extreme squares may have two different connected exclusive regions.
\end{claim}
\begin{proof}
Consider any non-extreme square $s_i$. The square $s$ has a square $s_{i-1}$ to its left and a square $s_{i+1}$ to its right. There are $4$ cases.\\
i) All three squares are in ASC order. Then the exclusive regions of $s_{i}$ must lie around its top left corner and/or its bottom right corner.\\
ii) All three squares are in DESC order. Then the exclusive regions of $s_{i}$ must lie around its top right corner and/or its bottom left corner.\\
iii) $s_i$ is below both $s_{i-1}$ and $s_{i+1}$: Then $s_{i}$ has only one exclusive region which is either rectangular or $L$-shaped.\\
iv) $s_i$ is above both $s_{i-1}$ and $s_{i+1}$: Then $s_{i}$ has only one exclusive region which is either rectangular or $L$-shaped.\\
Since the horizontal slab has height $1$, hence there can be at most two squares having two different connected exclusive regions. Specifically, in a monotonic DESC clique, the rightmost square intersecting $L_2$ and the leftmost square intersecting $L_1$. And in a monotonic ASC clique, the rightmost square intersecting $L_1$ and the leftmost square intersecting $L_2$.
\end{proof}
Since all the squares in our solution are necessary, hence for every square $s\in Sol$, there exists a set of points $Excl(s)$ such that the points in $Excl(s)$ are contained exclusively in $s$ and no other square in $Sol$. These points in $Excl(s)$ are called exclusive points to $s$. We make the following crucial claim about exclusive points.
\begin{claim}\label{claim:excl}
Let $s_1, s_2$ be two consecutive squares in a maximum clique of $Sol$ such that $s_1\prec s_2$ and $s_1$ is not the leftmost square in $Sol$. No input square $s$ can contain all the points in $Excl(s_1)\cup Excl(s_2)$.
\end{claim}
\begin{proof}
We have already established that any clique in our solution $Sol$ containing no redundant squares can be of only a few types. There are $6$ possibilities for the consecutive squares $s_1$ and $s_2$. We analyze them below. In each of the cases below, assume for the sake of contradiction, that there exists a square $s$ such that $s$ covers all the points in $Excl(s_1)\cup Excl(s_2)$.
\begin{figure}[ht!]
    \centering
    \includegraphics[width=10cm]{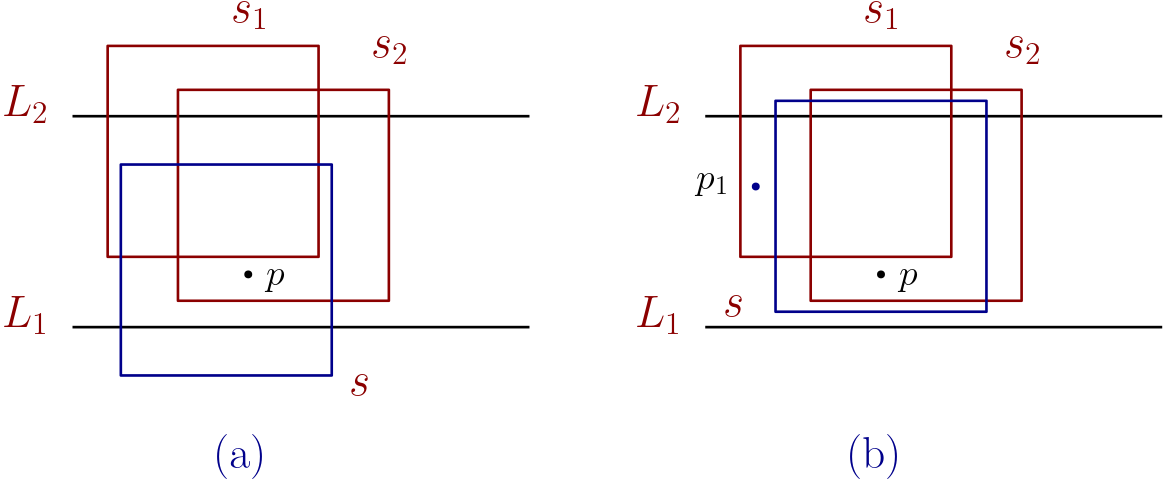}
    \caption{The squares $s_1$ and $s_2$ are top anchored and are in descending order in the clique under consideration.}
    \label{fig:1ab_top_desc}
\end{figure}

\begin{enumerate}
    \item Top Anchored DESC: Here $s_1$ and $s_2$ are intersecting the top line $L_2$ and $s_1\prec s_2$ and $s_1$ lies above $s_2$. Again there are two subcases.
    \begin{enumerate}
        \item If $s$ intersects the bottom line $L_1$: Then at the leftmost exclusive point $p$ of $s_2$, our algorithm has to make a choice between $s$ and $s_2$. Since $s_1$ is already picked due to Claim \ref{claim:placement}, our algorithm will prefer $s$ to $s_2$ since choosing $s_2$ gives a floating clique. And our algorithm would never pick $s_2$ in the future again. Recall that our greedy algorithm prefers floating cliques to anchored cliques. Refer to Figure \ref{fig:1ab_top_desc}(a) for an illustration.
        \item If $s$ intersects the top line $L_2$: Then at the leftmost exclusive point $p$ of $s_2$, our algorithm would have picked $s$ instead of $s_2$ since picking $s$ would render $s_1$ redundant. Since the exclusive region of $s_1$ is rectangular, and the square $s$ covers all the points in $Excl(s_1)$, hence $s$ lies to the left of $s_2$, i.e., $s \prec s_2$. Therefore, $s$ covers every point in $s_1\cap s_2$ lying to the left of $p$. Consider a point $p_1\in s_1\setminus s_2$, which is covered by another square $s_0\in Sol$ but presumably not covered by $s$. Clearly, all the exclusive points of $s_0$ lie to the left of $p$. By Claim \ref{claim:placement}, $s_0$ must have been picked by our solution already. Consider a point $p_2\in s_2\setminus s_1$, which is covered by another square $s_0\in Sol$ but presumably not covered by $s$. If $s_2 \prec s_0$ then all the points of $s_0$ lie to the right of $p$. 
        Hence, we need not worry about covering $p_2$ at this stage. Else if $s_0 \prec s_1$ then $s_0\cap s_2$ will be contained in $s_1\cap s_2$ and such a point $p_2$ cannot exist. 
        Thus picking $s$ during processing $p$ does not cause an increase in the active ply and our algorithm will pick $s$ greedily. Refer to Figure \ref{fig:1ab_top_desc}(b) for an illustration.
    \end{enumerate}
    
    \begin{figure}[ht!]
    \centering
    \includegraphics[width=10cm]{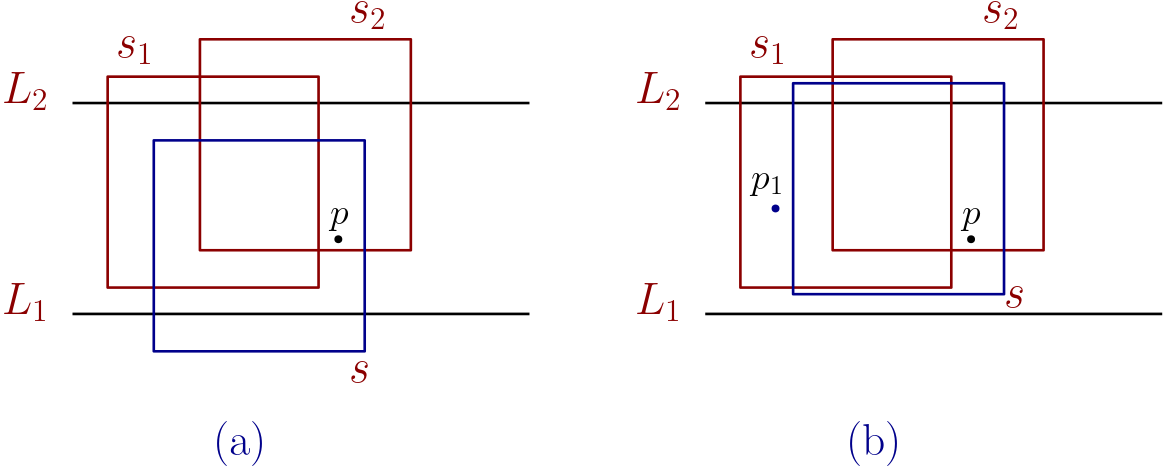}
    \caption{The squares $s_1$ and $s_2$ are top anchored and are in ascending order in the clique under consideration.}
    \label{fig:2ab_top_asc}
    \end{figure}
    
    \item Top Anchored ASC: Here $s_1$ and $s_2$ are intersecting the top line $L_2$ and $s_1\prec s_2$ and $s_1$ lies above $s_2$.  Again there are two subcases.
    \begin{enumerate}
        \item If $s$ intersects the bottom line $L_1$: Then at the leftmost exclusive point $p$ of $s_2$, our algorithm would have picked $s$ instead of $s_2$. The reason is exactly same as the argument for the case 1(a) above.
        \item If $s$ intersects the top line $L_2$: While processing the leftmost exclusive point $p$ of $s_2$, our algorithm has already picked $s_1$ since the rightmost exclusive point of $s_1$ must lie to the left of $p$. At $p$ our algorithm would have picked $s$ instead of $s_2$ since picking $s$ would also render $s_1$ redundant. 
        There are two possibilities. First, if $s_2\prec s$, then our algorithm would prefer $s$ to $s_2$ as it would give a narrower clique of same size. Second, if $s_1\prec s\prec s_2$, then $s$ covers every point in $s_1\cap s_2$. Consider a point $p_1\in s_1\setminus s_2$, which is covered by another square $s_0\in Sol$ but presumably not covered by $s$. If $s_0\prec s_1$, then $s_0$ is already picked by our algorithm when we are processing $p$. If $s_2\prec s_0$, then such a $p$ cannot exist. Refer to Figure \ref{fig:2ab_top_asc}. Consider a point $p_2\in s_2\setminus s_1$, which is covered by another square $s_0\in Sol$ but presumably not covered by $s$. If $s_2 \prec s_0$ then $s$ covers $p_2$ as $s$ covers $p$. Else if $s_0 \prec s_1$ then such a point $p_2$ cannot exist. Thus picking $s$ during processing $p$ does not cause an increase in the active ply and our algorithm will pick $s$ greedily.
    \end{enumerate}
    \item Bottom Anchored DESC: $s_1$ and $s_2$ are intersecting the bottom line $L_1$ and $s_1\prec s_2$ and $s_1$ is above $s_2$. Again there are two subcases.
    \begin{enumerate}
        \item If $s$ intersects the top line $L_2$: Then at the leftmost exclusive point $p$ of $s_2$, our algorithm would have picked $s$ instead of $s_2$. The reason is exactly same as the argument for the case 1(a) above.
        \item If $s$ intersects the bottom line $L_1$: While processing the leftmost exclusive point $p$ of $s_2$, our algorithm has already picked $s_1$ since the rightmost exclusive point of $s_1$ must lie to the left of $p$. At $p$, our algorithm would have picked $s$ instead of $s_2$ since picking $s$ would render $s_1$ redundant. 
        There are two possibilities. First, if $s_2\prec s$, then our algorithm would prefer $s$ to $s_2$ as it would give a narrower clique of same size. 
        Second, if $s_1\prec s\prec s_2$, then $s$ covers every point in $s_1\cap s_2$. Consider a point $p_1\in s_1\setminus s_2$, which is covered by another square $s_0\in Sol$. If $s_0\prec s_1$, then $s_0$ is already picked by our algorithm when we are processing $p$. If $s_2\prec s_0$, then such a $p$ cannot exist. Refer to Figure (). Consider a point $p_2\in s_2\setminus s_1$, which is covered by another square $s_0\in Sol$. If $s_2 \prec s_0$ then such a point $p_2$ lies to the right of $p$. Else if $s_0 \prec s_1$ then such a point $p_2$ cannot exist.        
        Thus picking $s$ for $p$ does not cause an increase in the active ply and this conforms to our greedy choice.
    \end{enumerate}
    
    \begin{figure}[ht!]
    \centering
    \includegraphics[width=10cm]{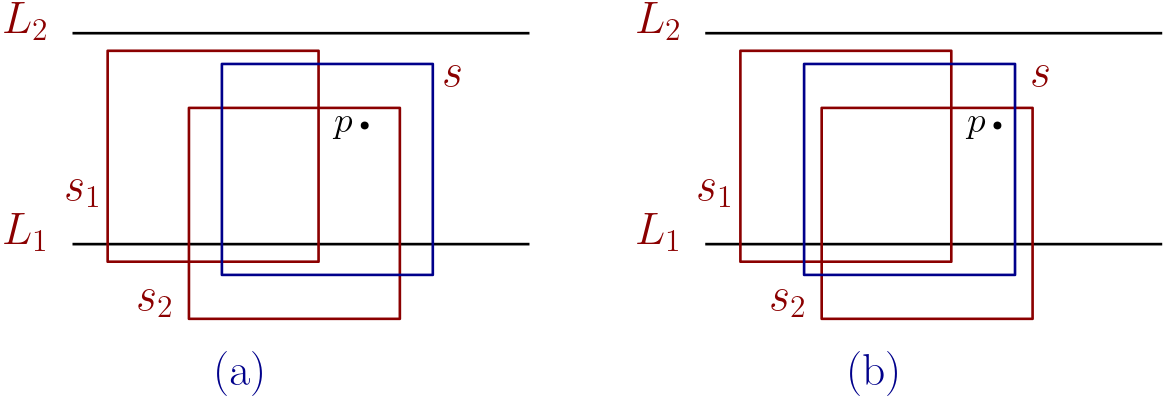}
    \caption{The squares $s_1$ and $s_2$ are bottom anchored and are in descending order in the clique under consideration.}
    \label{fig:3ab_top_asc}
    \end{figure}
    
    \item Bottom Anchored ASC: $s_1$ and $s_2$ are intersecting the bottom line $L_1$ and $s_1\prec s_2$ and $x_{T}(s_1)< x_{T}(s_2)$. 
    Again there are two subcases.
    \begin{enumerate}
        \item If $s$ intersects the top line $L_2$: Then at the leftmost exclusive point $p$ of $s_2$, our algorithm has to make a choice between $s$ and $s_2$. Since $s_1$ is already picked, our algorithm will prefer $s$ to $s_2$ since choosing $s_2$ gives a floating clique. And our algorithm would never pick $s_2$ in the future again. Recall that our greedy algorithm prefers floating cliques to anchored cliques.
        \item If $s$ intersects the bottom line $L_1$: Then at the leftmost exclusive point $p$ of $s_2$, our algorithm would have picked $s$ instead of $s_2$ since picking $s$ would render $s_1$ redundant. Since the exclusive region of $s_1$ is rectangular, and the square $s$ covers all the points in $Excl(s_2)$, hence $s$ lies to the left of $s_2$, i.e., $s \prec s_2$. Therefore, $s$ covers every point in $s_1\cap s_2$ lying to the left of $p$. Consider a point $p_1\in s_1\setminus s_2$, which is covered by another square $s_0\in Sol$. Clearly, all the exclusive points of $s_0$ lie to the left of $p$. By lemma (\ref{claim:placement}), $s_0$ must have been picked by our solution already. Consider a point $p_2\in s_2\setminus s_1$, which is covered by another square $s_0\in Sol$. If $s_2 \prec s_0$ then all the points of $s_0$ lie to the right of $p$. Hence, we need not worry about covering $p_2$ at this stage. Else if $s_0 \prec s_1$ then $s_0\cap s_2$ will be contained in $s_1\cap s_2$ and such a point $p_2$ cannot exist. Thus picking $s$ during processing $p$ does not cause an increase in the active ply and our algorithm will pick $s$ greedily.
    \end{enumerate}
    
    \item Floating DESC: $s_1$ intersects the top line $L_2$ and $s_2$ intersects the bottom line $L_1$ and $s_1\prec s_2$. Again there are two subcases.
    \begin{enumerate}
        \item The square $s$ intersects the top line $L_1$:  If $s_2 \prec s$, then at the leftmost exclusive point $p$ of $s_1$, our algorithm would have picked $s$ instead of $s_1$ since picking $s$ also leads to a narrower clique of same size. If $s_1\prec s \prec s_2$, then at the leftmost exclusive point $p$ of $s_1$, our algorithm would have picked $s$ instead of $s_1$ since picking $s$ leads to a narrower clique of same size. If there exists a point $p_1\in s_1\cap s_2$ which is not covered by any other square in $Sol$, then while processing $p_1$, our algorithm will pick $s_2$ instead of $s_1$ since it leads to a narrower clique. Thus our algorithm never picks $s_1$. If $s\prec s_1 \prec s_2$, then at the leftmost exclusive point $p$ of $s_2$, our algorithm would have picked $s$ instead of $s_2$ and discarded $s_1$. If there exists a point $p_1\in s_1\cap s_2$ which is not covered by any other square in $Sol$ or by $s$, then while processing $p_1$, our algorithm will pick $s_2$ instead of $s_1$. Thus our algorithm never picks $s_1$.
        \item The square $s$ intersects the bottom line $L_1$: Exact same arguments as in Case 5(a) are applicable.
    \end{enumerate}
    \item Floating ASC: $s_1$ intersects the bottom line $L_1$ and $s_2$ intersects the top line $L_2$ and $s_1\prec s_2$. Similar arguments as presented in the Floating DESC case apply to this case.
\end{enumerate}
Since there are no other possibilities for $s_1, s_2$ and $s$, this completes the proof of our claim.
\end{proof}
\begin{lemma}
Consider one of the maximum cliques, say $K$, in our solution $Sol$. To cover the exclusive points of the squares forming $K$, any feasible set cover has to pick $\lfloor k/3 \rfloor$ squares where $|K|=k$.
\end{lemma}
\begin{proof}
The necessity of $\lfloor k/3\rfloor$ squares to cover the exclusive points of the squares in the clique $K$ of $Sol$ is a direct consequence of Claim (\ref{claim:excl}).
We have already shown that the exclusive points in a ASC clique (resp. DESC clique) are monotonically ascending (resp. descending) from left to right except possibly at $4$ squares. First we argue for the case when $K$ is a monotonic clique of ASC type. Consider four consecutive squares in $K$, say, $s_1, s_2, s_3$ and $s_4$, all of which intersect the top line $L_2$. No square $s$ can cover exclusive points from all the $4$ squares. Otherwise such a square $s$ would end up covering all points in $Excl(s_1)\cup Excl(s_2)$. Refer to Figure \ref{fig:a_mandatory_squares}. Similar argument is applicable if the squares are all bottom-intersecting. The only exception can take place if some squares are top-intersecting and some are bottom-intersecting as shown in Figure \ref{fig:mandatory_squares}. In this case, a square $s$ may cover exclusive points from all of $s_1, s_2, s_3, s_4$. But covering exclusive points from $5$ squares will be impossible for similar reasons. \\
Now partition the squares of the maximum clique into groups of $3$ from left to right. For each such group at least $1$ square is necessary except possibly for one group at the transition from top-to-bottom or bottom-to-top. Similar arguments apply for transition squares if any. Hence $\lfloor k/3 \rfloor$ squares are necessary.
\end{proof}
\begin{figure}[ht!]
    \centering
    \includegraphics[width=10cm]{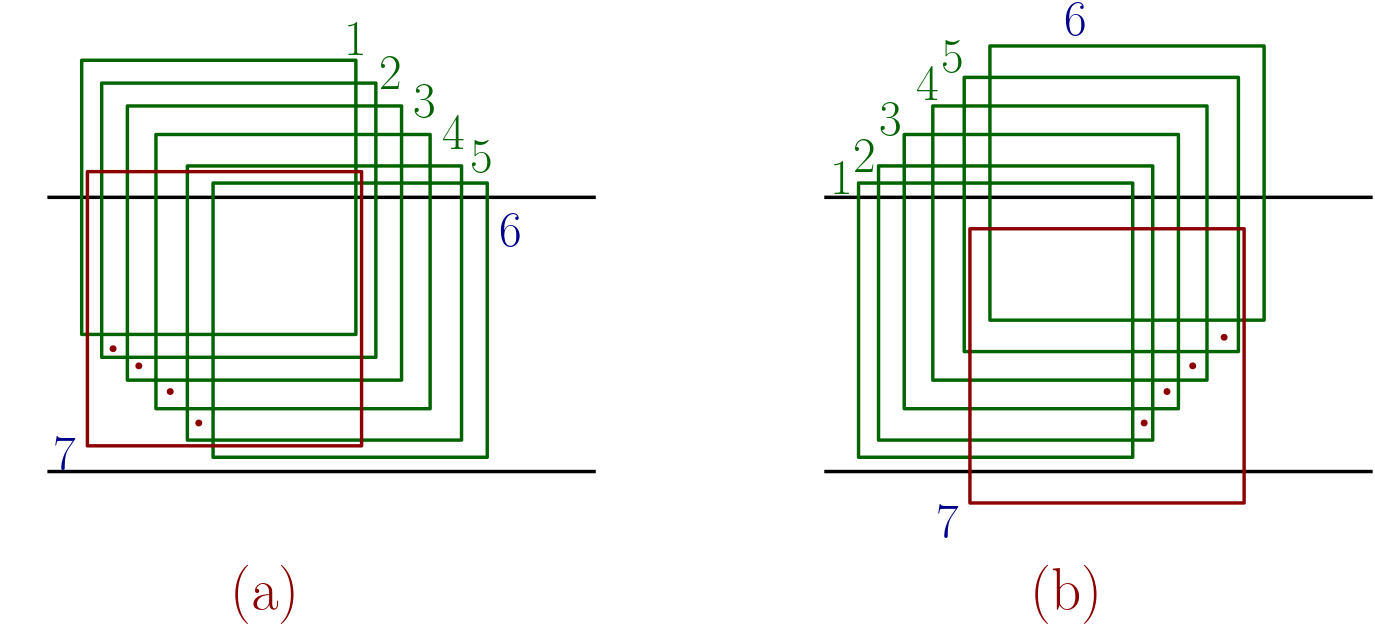}
    \caption{(a) Consider the 4 consecutive squares $2,3,4,5$ in this top-anchored monotonic descending clique. The square $7$ covers exclusive points from all the four squares. Specifically it covers all the points in $Excl(3)\cup Excl(4)$. (b) Consider the 4 consecutive squares $2,3,4,5$ in this top-anchored monotonic ascending clique. The square $7$ covers exclusive points from all the four squares. Specifically it covers all the points in $Excl(3)\cup Excl(4)$.}
    \label{fig:a_mandatory_squares}
\end{figure}

\begin{figure}[ht!]
    \centering
    \includegraphics[width=5cm]{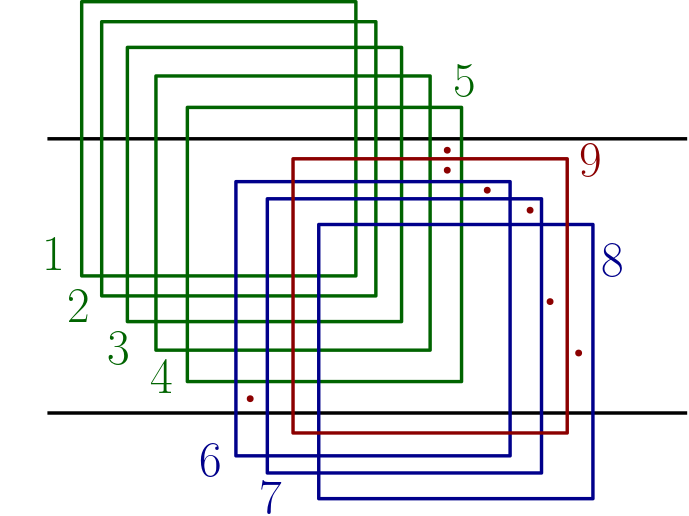}
    \caption{The square $9$ in this floating clique covers exclusive points from the squares $5, 6, 7, 8$ and covers all exclusive points of only one square, i.e., square $7$.}
    \label{fig:mandatory_squares}
\end{figure}

\begin{lemma}\label{lemma:9approx}
    Among the squares required to cover a clique of size $k$ in $Sol$, at least $\lfloor\frac{k}{9}\rfloor - 1$ squares have a common intersection.
\end{lemma}
\begin{proof}
First, consider a top-anchored ASC clique. The exclusive points are also monotonically ascending except possibly at the leftmost square. In  this clique, consider a non-extreme exclusive point $p$. Either $p$ can be covered by a top intersecting square from the left or it can be covered by a bottom intersecting square from the right as shown in the Figure \ref{fig:cliques_overlap}(b). This implies that all the bottom intersecting squares covering non-extreme exclusive points intersect. Similarly, all the top intersecting squares covering non-extreme exclusive points intersect. Hence, two cliques are formed. Applying the pigeonhole principle, one of the cliques is of size at least $\frac{\lfloor k/3\rfloor - 2}{2} = \lfloor\frac{k}{6}\rfloor -1$.

For a monotonic floating clique or a clique which is made up of two monotonic sequences of squares, i.e., a clique of type ASC$|$DESC or DESC$|$ASC, the exclusive points may have two different monotonic sequences. Therefore we need a different argument. 

Consider a clique $Q$ is of type bottom-anchored ASC$|$DESC. Observe that there are two monotonic sequences of exclusive points in $Q$ - the first sequence is ascending and the second sequence is descending. As before, two cliques are necessary to cover the points of the first monotonic sequence of exclusive points. Denote these cliques by $Q_1$ and $Q_2$ where $Q_1$ is top-anchored and $Q_2$ is bottom-anchored. Also, there are two other cliques in $OPT$ for covering the exclusive points of the second monotonic sequence. Denote these cliques by $Q_3$ and $Q_4$ where $Q_3$ is bottom-anchored and $Q_4$ is top-anchored. All the $L_1$ (i.e., bottom line) intersecting squares, i.e., the squares in $Q_2$ and $Q_3$ have a common intersection region. Refer to the Figure \ref{fig:cliques_overlap}(d). Therefore, in $OPT$ there is exists a clique of size
\[
\max(|Q_1|, |Q_2|+|Q_3|, |Q_4|)
\]
Applying the pigeonhole principle, one of the cliques have size at least $\frac{\lfloor k/3 \rfloor}{3} - 1 = \lfloor\frac{k}{9}\rfloor - 1$.

Consider a clique $Q$ is of type floating DESC$|$DESC. Observe that there are two monotonic sequences of exclusive points in $Q$ - both the second sequences are descending. As before, two cliques are necessary to cover the points of the first monotonic sequence of exclusive points. Denote these cliques by $Q_1$ and $Q_2$ where $Q_1$ is bottom-anchored and $Q_2$ is top-anchored. Also, there are two other cliques in $OPT$ for covering the exclusive points of the second monotonic sequence. Denote these cliques by $Q_3$ and $Q_4$ where $Q_3$ is bottom-anchored and $Q_4$ is top-anchored. Here the geometry is such that all the squares in $Q_2$ which are intersecting $L_2$ intersect the squares in $Q_3$, which intersect $L_1$. Refer to the Figure \ref{fig:cliques_overlap}(e). Therefore, in $OPT$ there is exists a clique of size
\[
\max(|Q_1|, |Q_2|+|Q_3|, |Q_4|)
\]
Applying the pigeonhole principle, one of the cliques have size at least $\frac{\lfloor k/3\rfloor}{3} - 1 = \lfloor\frac{k}{9}\rfloor - 1$.
\end{proof}

\begin{figure}[ht!]
    \centering
    \includegraphics[width=12cm]{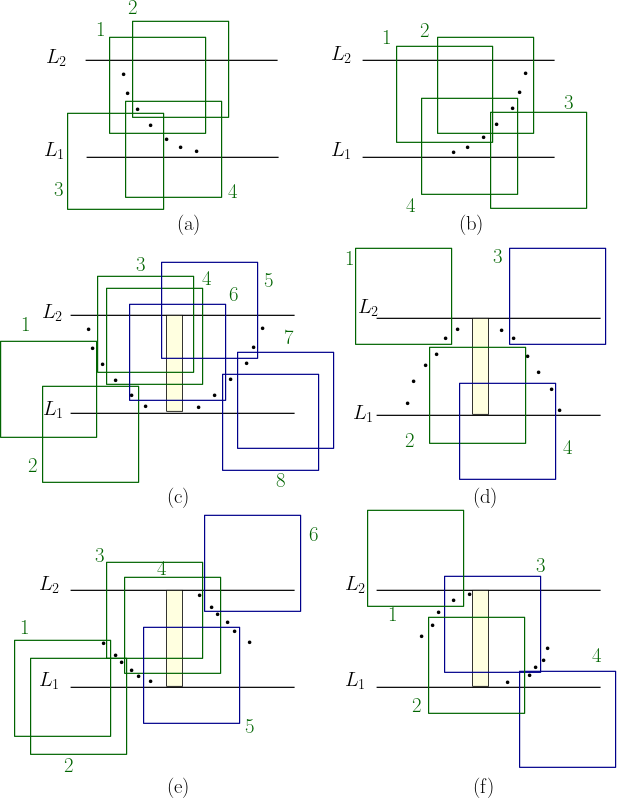}
    \caption{Shows the different cases for the proof of Lemma (\ref{lemma:9approx}).}
    \label{fig:cliques_overlap}
\end{figure}
\begin{theorem}
Given s set of $n$ points and $m$ axis-parallel unit squares on the plane, our algorithm computes a $(27+\epsilon)$-factor approximation of the minimum ply cover in $O((nm)^2)$ time, where $\epsilon>0$ is a small positive constant.
\end{theorem}
\begin{proof}
    The approximation factor is a direct consequence of Theorem \ref{thm:3times} and Lemma \ref{lemma:9approx}. The algorithm for the horizontal slab subproblem computes a table having $nm$ entries. Each entry can be computed in $O(m)$ time. Therefore, each subproblem requires $O(nm^2)$ time. There are at most $n$ subproblems. Hence the total time required is $O((nm)^2)$.
\end{proof}

\section{Conclusion}
In this paper we have given an algorithmic technique that runs fast for the minimum ply cover problem with axis-parallel unit squares. We have been able to characterize the structure of any clique in our solution and compare it with the maximum clique of the intersection graph of an optimal solution. It may be possible to improve the approximation ratio further. We believe that our
technique can be generalized to obtain polynomial-time approximation algorithms for broader class of objects.
\vspace{2pt}

\textbf{Acknowledgement}. I would like to thank Sathish Govindarajan for many useful discussions on the minimum ply covering problem and for his valuable comments. I would also like to thank Aniket Basu Roy and Shirish Gosavi for their time discussing the problem with me.

\bibliographystyle{plainnat}
\bibliography{ref}


\end{document}